\documentclass[11pt,a4paper]{article}
\usepackage{fullpage}

\usepackage[utf8]{inputenc}
\usepackage[T1]{fontenc}
\usepackage{lmodern}
\usepackage{hyperref}
\usepackage{url}
\usepackage{booktabs}
\usepackage{amsfonts}
\usepackage{nicefrac}
\usepackage{microtype}
\usepackage{xcolor}
\usepackage{enumitem}
\usepackage{wrapfig}

\usepackage{bm}
\usepackage{amsmath}
\usepackage{amssymb}
\usepackage{amsthm}
\usepackage{bbm}
\usepackage{thm-restate}
\usepackage{cite}
\usepackage[vlined,linesnumbered,ruled]{algorithm2e}
\usepackage{comment}
\usepackage{tikz}
\usepackage{bbm}
\usepackage{caption}
\usepackage{subcaption}
\usepackage{multirow}

\usepackage[capitalize]{cleveref}

\newcommand{\X}{\bm{\mathcal{X}}}
\newcommand{\R}{\mathbb{R}}
\newcommand{\Z}{\mathbb{Z}}
\newcommand{\x}{\bm{x}}
\newcommand{\E}{\mathbb{E}}
\newcommand{\cost}{\ensuremath{\mathrm{cost}}}
\newcommand{\costinc}{\ensuremath{\mathrm{cost\mbox{-}increase}}}
\newcommand{\activecut}{\ensuremath{\mathrm{active}}}
\newcommand{\U}{\mathcal{U}}
\DeclareMathOperator{\poly}{poly}
\DeclareMathOperator{\polylog}{polylog}

\newtheorem{theorem}{Theorem}
\newtheorem{lemma}{Lemma}

\newcommand{\appref}[1]{Appendix~\ref{#1}}

\title{Nearly-Tight and Oblivious Algorithms\\ for Explainable Clustering}

\author{Buddhima Gamlath \and Xinrui Jia \and Adam Polak \and Ola Svensson}

\date{EPFL}

\begin{document}

\maketitle

\begin{abstract}
We study the problem of explainable clustering in the setting first formalized by Dasgupta, Frost, Moshkovitz, and Rashtchian (ICML 2020).
A $k$-clustering is said to be explainable if it is given by a decision tree where each internal node splits data points with a threshold cut in a single dimension (feature), and each of the $k$ leaves corresponds to a cluster.
We give an algorithm that outputs an explainable clustering that loses at most a factor of $O(\log^2 k)$ compared to an optimal (not necessarily explainable) clustering for the $k$-medians objective, and a factor of $O(k \log^2 k)$ for the $k$-means objective.
This improves over the previous best upper bounds of $O(k)$ and $O(k^2)$, respectively,
and nearly matches the previous $\Omega(\log k)$ lower bound for $k$-medians and our new $\Omega(k)$ lower bound for $k$-means.
The algorithm is remarkably simple. In particular, given an initial not necessarily explainable clustering in $\R^d$, it is oblivious to the data points and runs in time $O(dk \log^2 k)$, independent of the number of data points $n$.
Our upper and lower bounds also generalize to objectives given by higher $\ell_p$-norms.
\end{abstract}

\section{Introduction}

An important topic in current machine learning research is understanding how models actually make their decisions. For a recent overview on the subject of explainability and interpretability, see, e.g.,~\cite{molnar,murdoch}.
Many good methods exist (e.g.~\cite{ribeiro}) for interpreting black-box models, so called \emph{post-modeling explainability}, but this approach has been criticized~\cite{rudin} for providing little insight into the data. Currently, there is a shift towards designing models that are interpretable by design.

Clustering is a fundamental problem in unsupervised learning. A common approach to clustering is to minimize the $k$-medians or $k$-means objectives, e.g., with the celebrated Lloyd's~\cite{Lloyd82} or $k$-means++~\cite{Arthur07} algorithms. Both objectives are also widely studied from a theoretical perspective, and, in particular, they admit constant-factor approximation algorithms running in polynomial time~\cite{Charikar02,Byrka17,Kanungo04,Ahmadian20}.

In their recent paper~\cite{moshkovitz}, Dasgupta et al.~were the first to study provable guarantees for explainable clustering. They define a $k$-clustering to be \emph{explainable} if it is given by a decision tree, where each internal node splits data points with a \emph{threshold cut} in a single dimension (feature), and each of the $k$ leaves corresponds to a cluster (see Figure~\ref{fig:definition}).

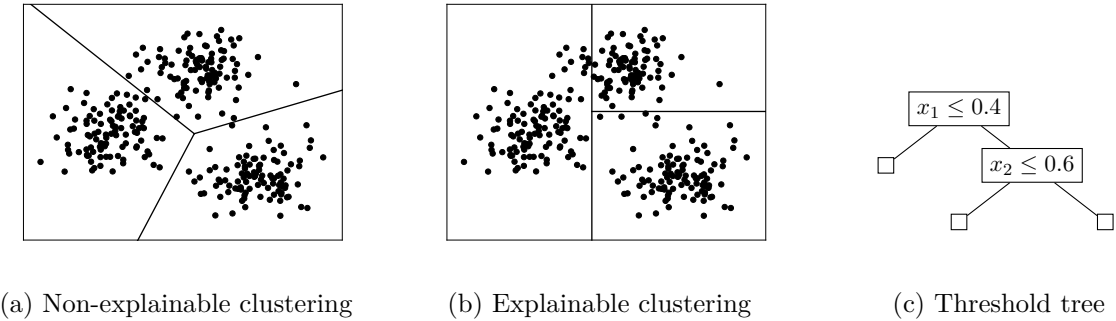
\begin{figure}
    \centering
    \begin{subfigure}[b]{0.34\textwidth}
        \includegraphics[width=\textwidth]{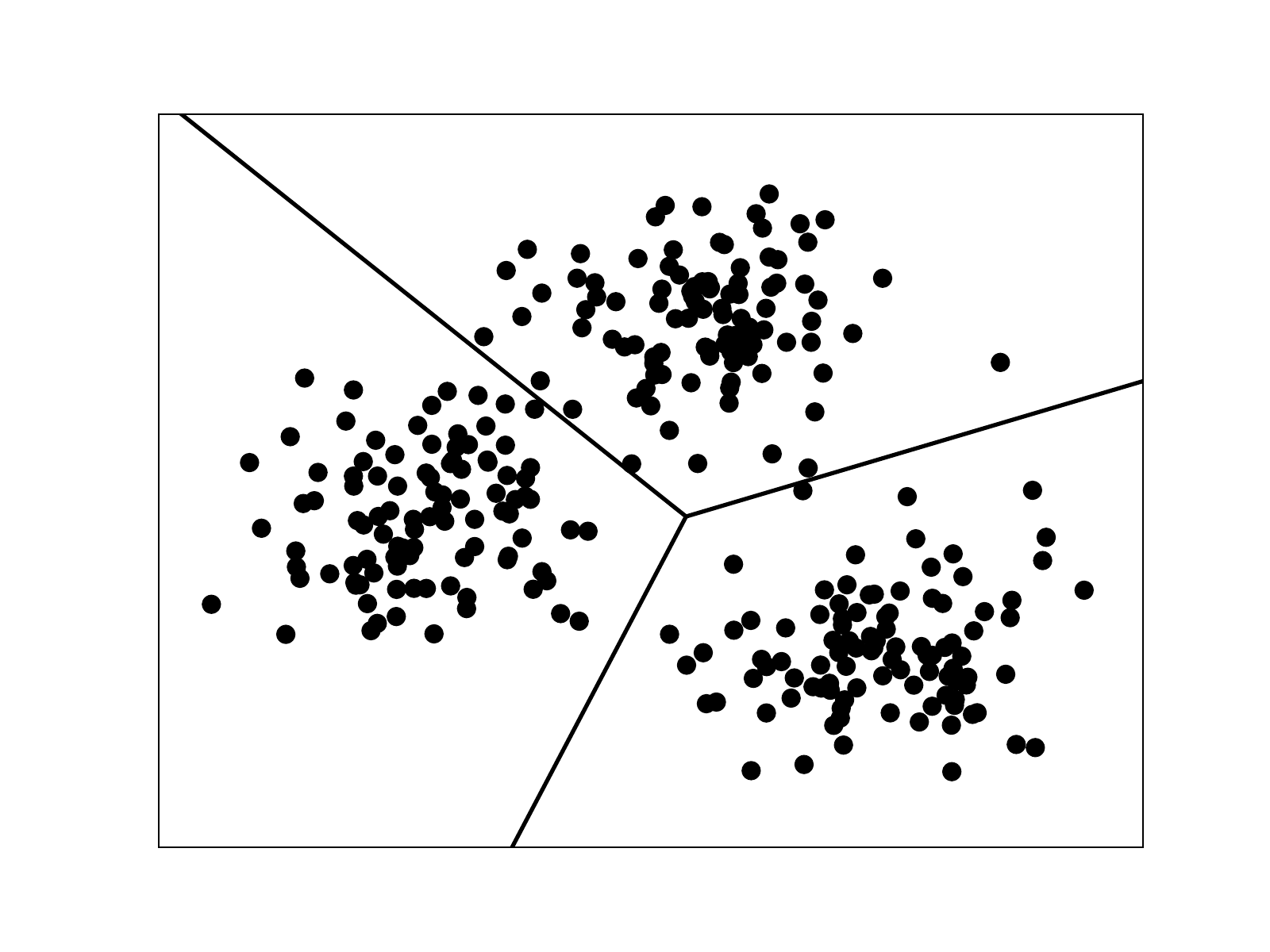}
        \caption{Non-explainable clustering}
    \end{subfigure}
    \hfill
    \begin{subfigure}[b]{0.34\textwidth}
        \includegraphics[width=\textwidth]{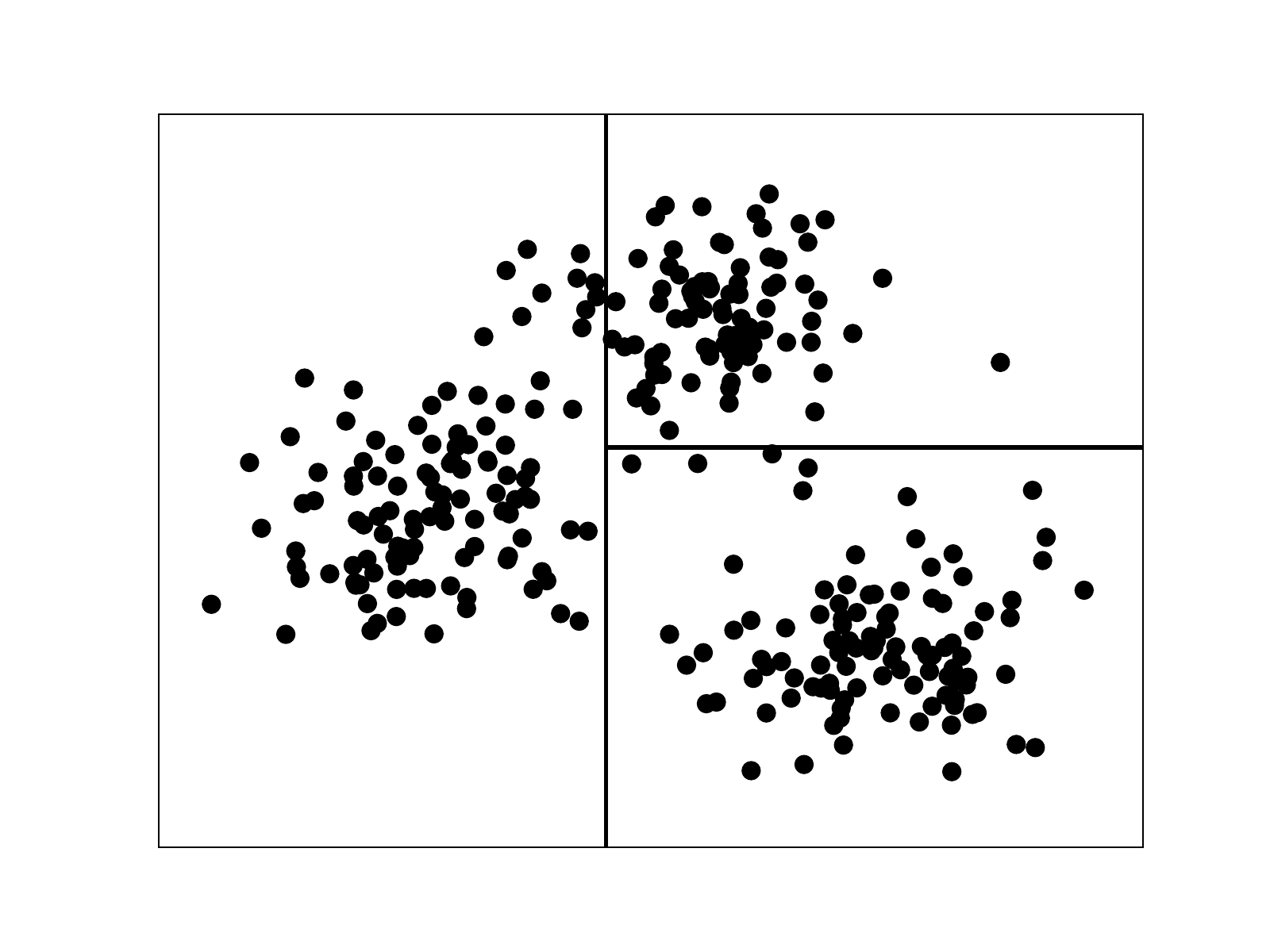}
        \caption{Explainable clustering}
    \end{subfigure}
    \hfill
    \begin{subfigure}[b]{0.3\textwidth}
        \centering
        \begin{tikzpicture}[sibling distance=10em,
            every node/.style = {shape=rectangle,
                draw, align=center,
                top color=white, bottom color=white, scale=0.8}, scale=0.5]
        \node {$x_1 \leq 0.4$}
          child { node {} }
          child { node {$x_2 \leq 0.6$}
            child { node {} }
            child { node {} }
        };
        \end{tikzpicture}
        \vspace{1.5em}
        \caption{Threshold tree}
    \end{subfigure}
    \caption{Examples of an optimal non-explainable and a costlier explainable clustering of the same set of points in $\R^2$, together with the threshold tree defining the explainable clustering.}
    \label{fig:definition}
\end{figure}

This definition is motivated by the desire to have a concise and easy-to-explain reasoning behind how the model chooses data points that form a cluster. See the original paper~\cite{moshkovitz} for an extensive discussion of motivations and a survey of previous (empirical) approaches to explainable clustering.

The central question to study in this setting is that of the \emph{price of explainability}: How much do we have to lose -- in terms of a given objective, e.g., $k$-medians or $k$-means -- compared to an optimal unconstrained clustering, if we insist on an explainable clustering, and can we efficiently construct such a clustering?

Dasgupta et al.~\cite{moshkovitz} proposed an algorithm that, given an unconstrained (non-explainable) \emph{reference clustering}\footnote{A reference clustering can be obtained, e.g., by running a constant-factor approximation algorithm for a given objective function. Then, the asymptotic upper bounds of the explainable clustering cost compared to the reference clustering translate identically to the bounds when compared to an optimal clustering.}, produces an explainable clustering losing at most a multiplicative factor of $O(k)$ for the $k$-medians objective and $O(k^2)$ for $k$-means, compared to the reference clustering. They also gave a lower bound showing that an $\Omega(\log k)$ loss is unavoidable, both for the $k$-medians and $k$-means objective. Later, Laber and Murtinho~\cite{laber} improved over the upper bounds in a low-dimensional regime $d \le k / \log(k)$, giving an $O(d \log k)$-approximation algorithm for explainable $k$-medians and an $O(d k\log k)$-approximation algorithm for explainable $k$-means.

\subsection{Our contributions}

\paragraph{Improved clustering cost.}
We present a randomized algorithm that, given $k$ centers defining a reference clustering and a number $p \geq 1$, constructs a threshold tree that defines an explainable clustering that is, in expectation, worse than the reference clustering by at most a factor of $O(k^{p-1} \log^2 k)$ for the objective given by the $\ell_p$-norm. That is $O(\log^2 k)$ for $k$-medians and $O(k \log^2 k)$ for $k$-means.

\paragraph{Simple and oblivious algorithm.}
Our algorithm is remarkably simple. It samples threshold cuts uniformly at random (for $k$-medians; $k$-means and higher $\ell_p$-norms need slightly fancier distributions) until all centers are separated from each other. In particular, the input to the algorithm includes only the centers of a reference clustering and not the data points.

As a consequence, the algorithm cannot overfit the data (any more than the reference clustering possibly already does), and the same expected cost guarantees hold for any future data points not known at the time of the clustering construction. Besides, the algorithm is fast; its running time does not depend on the number of data points $n$. A naive implementation runs in time $O(dk^2)$, and in Section~\ref{sec:implementation}, we show how to improve it to $O(dk \log^2 k)$ time, which is near-linear in the input size $dk$ of the $k$ reference centers.

\paragraph{Nearly-tight bounds.} We complement our results with a lower bound. We show how to construct instances of the clustering problem such that any explainable clustering must be at least $\Omega(k^{p-1})$ times worse than an optimal clustering for the $\ell_p$-norm objective. In particular, this improves the previous $\Omega(\log k)$ lower bound for $k$-means~\cite{moshkovitz} to $\Omega(k)$ .

In consequence, we give a nearly-tight answer to the question of the price of explainability. We leave a $\log(k)$ gap for $k$-medians, and a $\log^2(k)$ gap for $k$-means and higher $\ell_p$-norm objectives. See Table~\ref{tab:bounds} for a summary of the upper and lower bounds discussed above and recent independent works discussed in Section~\ref{sec:independent}.

\begin{table}
\caption{\textbf{Algorithms and lower bounds for explainable $k$-clustering in $\R^d$.} For a given objective function, how large a multiplicative factor do we have to lose, compared to an optimal unconstrained clustering, if we insist on an explainable clustering?}
\label{tab:bounds}
\centering
\renewcommand{\arraystretch}{1.3}
\begin{tabular}{@{}l@{\qquad}lllr@{}}
\toprule
  & $k$-medians & $k$-means & $\ell_p$-norm & \\
\midrule
\multirow{6}{*}{\rotatebox[origin=c]{90}{Algorithms}}
  & $O(k)$ & $O(k^2)$ & & Dasgupta et al.~\cite{moshkovitz} \\
  & $O(d \log k)$ & $O(kd \log k)$ & & Laber and Murtinho \cite{laber} \\
  & $O(\log^2 k)$  & $O(k \log^2 k)$ & $O(k^{p-1} \log^2 k)$ & \textbf{This paper} \\
  & $O(\log k \log \log k)$ & $O(k \log k \log \log k)$ & & Makarychev and Shan~\cite{makarychev} \\ 
  & $O(\log k \log \log k)$ & $O(k \log k)$ & & Esfandiari et al.~\cite{esfandiari} \\
  & $O(d \log^2 d)$ & & &  Esfandiari et al.~\cite{esfandiari} \\
  & & $O(k^{1-\nicefrac{2}{d}} \polylog k)$ & & Charikar and Hu~\cite{charikar2021nearoptimal} \\
\midrule
\multirow{5}{*}{\rotatebox[origin=c]{90}{Lower bounds}}
  & $\Omega(\log k)$ & $\Omega(\log k)$ & & Dasgupta et al.~\cite{moshkovitz} \\
  & & $\Omega(k)$ & $\Omega(k^{p-1})$ & \textbf{This paper} \\
  & & $\Omega(\nicefrac{k}{\log k})$ & & Makarychev and Shan~\cite{makarychev} \\ 
  & $\Omega(\min(d, \log k))$ & $\Omega(k)$ & & Esfandiari et al.~\cite{esfandiari} \\
  & & $\Omega(k^{1-\nicefrac{2}{d}} / \polylog k)$ & & Charikar and Hu~\cite{charikar2021nearoptimal} \\
\bottomrule
\end{tabular}
\end{table}

\subsection{Technical overview}
\label{overview}

The theoretical guarantees obtained by Dasgupta et al.~\cite{moshkovitz} depend on the number of clusters $k$ and the height of the threshold tree obtained $H$. Their algorithm loses, compared to the input reference clustering, an $O(H)$ factor for the $k$-medians cost and $O(Hk)$ for $k$-means. 
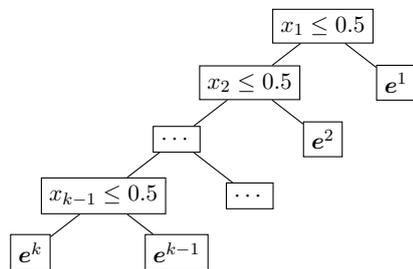
\begin{wrapfigure}{r}{0.4\textwidth}
  \centering
  \begin{tikzpicture}[sibling distance=10em, every node/.style = {shape=rectangle, draw, align=center, top color=white, bottom color=white, scale=0.8}, scale=0.5]
    \node {$x_1 \leq 0.5$}
      child { node {$x_2 \leq 0.5$} 
          child { node {$\cdots$}
              child { node {$x_{k-1} \leq 0.5$}
              child { node {$\bm{e}^k$}}
              child { node {$\bm{e}^{k-1}$}}
              }
              child { node {$\cdots$}}
              }
          child { node {$\bm{e}^2$}}
          }
      child { node {$\bm{e}^1$}};
  \end{tikzpicture}
  \caption{An optimal threshold tree for the $k$ standard basis vectors in $\R^{k}$. Any optimal threshold tree on this data set has height $k-1$.}
  \label{fig:basis_vectors}
\end{wrapfigure}
These approximations are achieved by selecting a threshold cut that separates some two centers and minimizes the number of points that get separated from their centers in the reference clustering.
This creates two children of a tree node, and the threshold tree is created by recursing on each of the children.
The height of the tree $H$ may need to be $k-1$. For example, consider the data set in $\R^{k}$ consisting of the $k$ standard basis vectors (see Figure~\ref{fig:basis_vectors}).
Laber and Murtinho~\cite{laber} replace the dependence on $H$ with $d$, the dimension of the data set, by first constructing optimal search trees for each dimension and then carefully using them to guide the construction of the threshold tree.
In our work, we obtain improved guarantees by using randomized cuts that are oblivious to the data points and depend only on the reference centers, in contrast to the above-mentioned two prior approaches, which selected cuts based on the data points.

There are two components to achieving our improved guarantees that correspond with two aspects of the minimum cut algorithm of~\cite{moshkovitz}: the use of the minimum cut, and the height of the threshold tree produced.
The first observation is that, for the $\ell_1$-norm, we do not lose in the analysis by taking a cut uniformly at random compared to always using the minimum cut.
(The corresponding distribution for higher $\ell_p$-norms is proportional to the $p$-th power of the distance to the closest center.)
Indeed, using a random cut makes us robust against specifically engineered examples, such as the one that fools the minimum cut algorithm of~\cite{moshkovitz} 
(see \appref{example-k}).
In that example we add dimensions in which a cut is minimum, but these minimum cuts produce a tree of height $\Omega(k)$ whose cost is $\Omega(k)$ times larger than the optimum.

However, threshold trees of height $\Omega(k)$ are unavoidable in certain instances as seen in the example with $k$ standard basis vectors (Figure~\ref{fig:basis_vectors}).
This leads to our second observation that it is necessary to use a tighter upper bound on the cost of reassigning a point since any height $k-1$ threshold tree produced on this example is actually optimal. 
Using the diameter definition in~\cite{moshkovitz}, the cost of each cut is upper bounded by $k$ while the actual distance between any two centers is at most $2$, which is also a valid upper bound for the reassignment cost.
Hence, we use the maximum distance between any two centers to upper bound the cost of misclassifying a point.

\paragraph{Open problems.}
We conjecture that our $k$-medians algorithm is asymptotically optimal.
In particular, we believe the actual approximation ratio of our algorithm is $1 + H_{k-1}$, where $H_{n}$ is the $n$-th harmonic number (recall that $\ln(n) \leq H_n \leq 1+ \ln(n)$).
There are two potential barriers in our current analysis that prevent us from demonstrating this optimality.
The first is that our upper bound on the cost increase of assigning a single point to a wrong center is not tight, and secondly, our analysis may include the cost of the same point multiple times.
Despite the further developments mentioned in Section~\ref{sec:independent}, it still remains to fully resolve the correct asymptotic price of explainability.

Some potential directions to expanding our work include parallelizations, generalizing the notion of explainability, and defining natural clusterability assumptions under which the price of explainability is reduced. 
Constructing a threshold tree seems inherently sequential; it would be interesting to explore parallelizations for faster implementation.
Another direction would be to allow each node to be a hyperplane in a chosen number of dimensions instead of only splitting along one feature. Finally, it seems a non-trivial question to find a right clusterability assumption on the data points distribution -- that would allow us to overcome the existing lower bounds -- because these lower bounds are in fact very ``clusterable'' instances, in the traditional usage of this notion.

\subsection{Independent work}
\label{sec:independent}
We note independent further developments by Makarychev and Shan~\cite{makarychev}; Esfandiari, Mirrokni, and Narayanan~\cite{esfandiari}; and Charikar and Hu~\cite{charikar2021nearoptimal}.

Makarychev and Shan~\cite{makarychev} showed $O(\log k \log \log k)$ and $O(k \log k \log \log k)$ upper bounds for $k$-medians and $k$-means, respectively, thus improving over our bounds by a factor of $\nicefrac{\log k}{\log \log k}$. Their $k$-medians algorithm is essentially the same as our \emph{modified} Algorithm~\ref{alg:random} (see Section~\ref{sec:mod-l1-alg}), but they provide a tighter analysis. Their $k$-means upper bound follows from combining their $k$-medians algorithm with their insightful reduction from $k$-means to $k$-medians that loses an $O(k)$ factor. However, the $k$-means algorithm resulting from that combination is essentially the same as our Algorithm~\ref{alg:randomp}. They also provide an $\Omega(\nicefrac{k}{\log k})$ lower bound for $k$-means, which is slightly worse than ours. Finally, they study the explainable $k$-medoids problem (i.e., $k$-medians with $\ell_2$ norm), and provide an $O(\log^{\nicefrac{3}{2}} k)$ upper bound and an $\Omega(\log k)$ lower bound.

Esfandiari, Mirrokni, and Narayanan~\cite{esfandiari} also give an $O(\log k \log \log k)$ upper bound for $k$-medians. Their algorithm is essentially the same as our (unmodified) Algorithm~\ref{alg:random}, and, again, they provide a tighter analysis. Moreover, they prove that their algorithm gives an $O(d \log^2 d)$ guarantee, improving over the work of Laber and Murtinho~\cite{laber} for low-dimensional spaces. They also give an $O(k \log k)$ upper bound for $k$-means, improving over the result of Makarychev and Shan by a factor of $\log \log k$. Their $k$-means algorithm is similar to our Algorithm~\ref{alg:randomp} but samples cuts from a different distribution. They also match our $\Omega(k)$ lower bound for $k$-means, and improve the $k$-medians lower bound of Dasgupta et al.~\cite{moshkovitz} to $\Omega(\min(d, \log k))$.

Charikar and Hu~\cite{charikar2021nearoptimal} focus on explainable $k$-means and present an $O(k^{1-\nicefrac{2}{d}} \poly(d \log k))$-approximation algorithm, which is better than any previous algorithm when $d = O(\nicefrac{\log k}{\log \log k})$ (in particular, for any constant dimension $d$). Resorting to an $O(k \polylog k)$-approximation algorithm (e.g., \cite{makarychev}) when this is not the case, they obtain an $O(k^{1-\nicefrac{2}{d}} \polylog k)$ upper bound. They match it with an $\Omega(k^{1-\nicefrac{2}{d}} / \polylog k)$ lower bound, which is tight up to polylogarithmic factors.

\section{Preliminaries}
\label{prelim}

\newcommand{\defn}[1]{\textbf{\emph{#1}}}

Following the notation of~\cite{moshkovitz}, we use bold variables for vector values and corresponding non-bold indexed variables for scalar coordinates.
Intuitively, a clustering is explainable because the inclusion of a data point $\x = [x_1, \dots, x_d]$ to a particular cluster is ``easily explained'' by whether or not $\x$ satisfies a series of inequalities of the form $x_i \leq \theta$.
These inequalities are called \defn{threshold cuts}, defined by a coordinate $i \in [d]$ (denoting the set $\{1, 2, \dots, d\}$) and a threshold $\theta \in \R$.
More precisely, a \defn{threshold tree} is a binary tree where each non-leaf node is a threshold cut $(i, \theta)$ which assigns the point $\bm{x}$ of that node into the left child if $x_i \leq \theta$ and the right child otherwise.
A~clustering is \defn{explainable} if the clusters are in bijection to the leaves of a threshold tree with exactly $k$~leaves that started with all the data points at the root. 

Given a set of points $\X = \{\x^1, \x^2, \dots, \x^n\} \subseteq \R^d$ and its clustering $\{C^1, \dots, C^k\}$, $\bigcup_{j=1}^{k} C^j=\X$, the \defn{$k$-medians cost} of the clustering is defined in~\cite{moshkovitz} as
\[ \mbox{cost}_1(C^1, \dots, C^k)
 = \sum_{j=1}^k \min_{\bm{\mu}\in\R^d} \sum_{\bm{x} \in C^j} \| \bm{x} - \bm{\mu} \|_1
 = \sum_{j=1}^k \sum_{\bm{x} \in C^j} \| \bm{x} - \mbox{median}(C^j) \|_1.\]
The \defn{$k$-means} cost is defined analogously with the square of the $\ell_2$ distance of every point to mean($C^j$).

For a set of centers $\U = \{\bm{\mu}^1, \dots, \bm{\mu}^k\} \subseteq \R^d$, a non-explainable clustering $\{\widetilde{C}^1, \dots, \widetilde{C}^k\}$ of $\X$ is given by $\widetilde{C}^j=\{\bm{x} \in \X \mid \bm{\mu}^j = \operatornamewithlimits{arg\,min}_{\bm{\mu}\in\U} \| \bm{x} - \bm{\mu} \|_1\}$, and we write $\mbox{cost}_1(\U) = \mbox{cost}_1(\widetilde{C}^1, \dots, \widetilde{C}^k)$. Note that $\mbox{cost}_1(\U) = \sum_{\bm{x} \in \X} \min_{\bm{\mu} \in \U} \| \bm{x} - \bm{\mu} \|_1$.

Given a threshold tree $T$, the leaves of $T$ induce an explainable clustering $\{\widehat{C}^1, \dots, \widehat{C}^k\}$, and we write $\mbox{cost}_1(T) = \mbox{cost}_1(\widehat{C}^1, \dots, \widehat{C}^k)$. 
In the analyses, however, we often upper bound the cost of each explainable cluster $\widehat{C}^j$ using the corresponding reference center $\bm{\mu}^j$:
\[\mbox{cost}_1(T) \le \sum_{j=1}^k \sum_{\x \in \widehat{C}^j} \| \x - \bm{\mu}^j \|_1.\]
These may not be optimal center locations, yet we are still able to obtain guarantees that are polylog away from being tight.

We generalize the above to higher $\ell_p$-norms, $p \geq 1$, as follows
\[
\mbox{cost}_p(C^1, \dots, C^k)
 = \sum_{j=1}^k \min_{\bm{\mu}\in\R^d} \sum_{\bm{x} \in C^j} \| \bm{x} - \bm{\mu} \|_p^p,
 \qquad
\mbox{cost}_p(T) \le \sum_{j=1}^k \sum_{\x \in \widehat{C}^j} \| \x - \bm{\mu}^j \|_p^p.\]

\section{Explainable \texorpdfstring{$k$}{k}-medians clustering}
\label{ell-1}

\newcommand{\calB}{\mathcal{B}}
\newcommand{\parent}{\operatorname{parentNode}}
\newcommand{\nextnode}{\operatorname{lastNode}}
\newcommand{\threshcut}{\operatorname{thresholdCut}}

In this section we present our algorithm for explainable \(k\)-medians and its analysis. 
Recall that our algorithm is oblivious to the data points: It determines the threshold tree using only the center locations.
The algorithm simply samples a sequence of cuts until it defines a threshold tree with each center belonging to exactly one leaf.
In what follows, we elaborate on this process in detail.

\newcommand{\dimrange}{\operatorname{I}}
\newcommand{\allcuts}{\operatorname{AllCuts}}
The algorithm's input is a set of centers \(\mathcal{U} = \{\bm{\mu}^1, \bm{\mu}^2, \ldots, \bm{\mu}^k \} \subset \R^d \).
We consider cuts that intersect the bounding box of $\U$. 
Letting \( \dimrange_i = [\min_{j \in [k] }\mu^j_i, \max_{j \in [k]}\mu^j_i] \) be the interval between the minimum and maximum \(i\)-coordinate of centers,
the set of all possible cuts that intersect the bounding box of $\U$ is \(\allcuts = \left\{ (i, \theta) : i \in [d],\theta \in \dimrange_i \right\}. \) 
Our algorithm uses a stream of independent uniformly random cuts from \( \allcuts \). In particular, the probability density function of \((i,\theta) \in \allcuts \) is \(1/L\) where \(L = \sum_{i\in[d]}|I_i| \) is the sum of the side lengths of the bounding box of \(\mathcal{U}\).

The algorithm simply takes cuts from this stream until it produces a threshold tree.
To this end, it maintains a tentative set of tree leaves, each identified by a subset of centers, and continues until it has \(k\) leaves of singleton sets.
We say a cut \emph{splits} a leaf if the cut properly intersects with the bounding box of the corresponding subset of centers. 
In other words, a cut \((i,\theta)\) splits a leaf \(B\) if and only if the two sets \(B^- = \{\bm{\mu} \in B : {\mu}_i \leq \theta \} \) and \(B^+ = \{\bm{\mu} \in B : {\mu}_i > \theta \} \) are both non-empty.
At the beginning, the algorithm starts with a single leaf identified by $\mathcal{U}$, the set of all centers.
It then samples a cut \((i,\theta)\) and checks if it splits any existing leaf. 
If so, it saves the cut, and for each leaf \(B\) that gets split by the cut into \(B^{-} \) and \(B^{+} \), adds \(B^-\) and \(B^+\) as two new leaves rooted at \(B\).
These saved cuts define the output threshold tree.
We present the pseudo-code of this algorithm in \cref{alg:random}. 

\newcommand{\treeleaves}{\operatorname{Leaves}}
\begin{algorithm}[ht]
  \DontPrintSemicolon
   \textbf{Input: } A collection of \( k \) centers \(\mathcal{U} = \{\bm{\mu}^1, \bm{\mu}^2, \ldots, \bm{\mu}^k\} \subset \R^d \). \;
   \textbf{Output: } A threshold tree with \(k\) leaves. \;
   \(\treeleaves \gets \{ \mathcal{U} \}\) \; 
   \While{\( |\treeleaves| < k \) } { \label{line:while}
        Sample \((i, \theta) \) uniformly at random from \(\allcuts\). \label{line:sample} \;  
        \For{each \(B \in \treeleaves \) that are split by \( (i, \theta) \)} { \label{line:for}
            Split \(B\) into \(B^-\) and \(B^+\) and add them as left and right children of \(B\).\label{line:split} \; 
            Update \(\treeleaves\). \label{line:update_leaves}
        }
   }
   \textbf{return} the threshold tree defined by all cuts that separated some \( B \). \;
   \caption{Explainable $k$-medians algorithm.}
   \label{alg:random}
\end{algorithm}

\subsection{Cost analysis}
\label{sec:l1-cost-analysis}
We show that Algorithm~\ref{alg:random} satisfies the following guarantees.
\begin{theorem}
  Given reference centers $\mathcal{U} = \{\bm{\mu}^1, \bm{\mu}^2, \ldots, \bm{\mu}^k\}$, Algorithm~\ref{alg:random} outputs a threshold tree $T$ whose expected cost satisfies
    \[
      \E[\cost_1(T)] \leq O(\log(k) \cdot (1+ \log(c_{\max}/c_{\min}))) \cdot \cost_1(\mathcal{U})\,,
    \]
    where $c_{\max}$ and $c_{\min}$ denote the maximum and minimum pairwise distance between two centers in $\mathcal{U}$, respectively. Furthermore, with probability at least $1-1/k$, Algorithm~\ref{alg:random} samples a threshold tree $T$ from a distribution with $\E[\cost_1(T)] \leq O(\log^2 k) \cdot  \cost_1(\mathcal{U})$.
    \label{thm:l1-alg}
\end{theorem}

The (more involved) proof of the furthermore statement is discussed in Section~\ref{sec:mod-l1-alg} and formally proved in 
\appref{sec:omitted-sec-3}.
We remark that the success probability $1-1/k$ can be made larger by only slightly increasing the hidden constant in the cost guarantee. Furthermore, in Section~\ref{sec:mod-l1-alg}, we give a slight adaptation of the above algorithm that has an expected cost bounded by $O(\log^2 k) \cdot \cost_1(\mathcal{U})$. 
The remaining part of this section is devoted to proving the upper bound on the expected cost of Algorithm~\ref{alg:random}. 
\paragraph{Proof outline.}
First, in Lemma~\ref{lemma:l1-nr-cut}, we show that a random cut in expectation separates $\cost_1(\mathcal{U})/L$ points from their closest centers. Indeed, note that the probability of separating a point $\bm{x}$ from its center $\pi(\bm{x})$ is at most $\|\bm{x} - \pi(\bm{x})\|_1/L$, and on the other hand, $\cost_1(\mathcal{U})   = \sum_{x\in \X} \|\bm{x} - \pi(\bm{x})\|_1$, hence the bound follows from linearity of expectation. Each such separated point incurs a cost of at most $c_{\max}$. Next, in Lemma~\ref{lemma:l1-drop}, we show that with good probability $O(\log(k) \cdot L/c_{\max})$ random cuts separate all pairs of centers that are at distance at least $c_{\max} / 2$ from each other. Morally, the cost of halving $c_{\max}$, which we will need to perform $1+ \log(c_{\max}/c_{\min})$ many times, is therefore $\cost_1(\mathcal{U})/L \cdot c_{\max} \cdot O(\log(k) \cdot L/c_{\max}) = O(\log(k)) \cdot  \cost_1(\mathcal{U})$, and the bound follows (see Lemma~\ref{lemma:cost-increase}).

\paragraph{Formal analysis of the expected cost.}
  We first  bound the number of points that are separated from their closest center by a random cut. This quantity is important as it upper bounds the number of points whose cost is increased in the final tree due to the considered cut. Recall that $L = \sum_{i=1}^d |I_i|$ denotes the total side lengths of the bounding box of the input  centers $\mathcal{U}$. We also let $f_i(\theta)$ be the number of points separated from their closest center by  the cut $(i, \theta)$. 

\begin{lemma}
    We have $\E_{(i,\theta)}[f_i(\theta)] \leq \cost_1(\mathcal{U})/L$ where the expectation is over a uniformly random threshold cut  $(i,\theta) \in \allcuts$.
    \label{lemma:l1-nr-cut}
\end{lemma}
\begin{proof}
  For a point $\bm{x}\in \X$ let $\pi(\bm{x})$ denote the closest center in $\mathcal{U}$. Then by definition,
  \[
   \cost_1(\mathcal{U})   = \sum_{x\in \X} \|\bm{x} - \pi(\bm{x})\|_1   = \sum_{i=1}^d \sum_{\bm{x}\in X}  |{x}_i - \pi(\bm{x})_i|\,.
  \]
 Moreover, if we let $f_i(\theta)$ be the number of points separated from their closest center by  the cut $(i, \theta)$, we can rewrite the cost of a fixed dimension $i$ as follows:
 \[
     \sum_{\bm{x}\in X}  |{x}_{i} - \pi(\bm{x})_{i}| = \sum_{\bm{x}\in X}  \int_{-\infty}^\infty \mathbbm{1}[\text{$\theta$ between  ${x}_{i}$ and $\pi(\bm{x})_{i}$}] d\theta = \int_{-\infty}^\infty f_{i}(\theta) d\theta\,.
 \]
We thus have $\cost_1(\mathcal{U}) = \sum_{i=1}^d \int_{-\infty}^\infty f_{i}(\theta) d\theta$.
 
 At the same time,  if we let $[a_i, b_i]$ denote the interval $I_i$, then $\frac{1}{|I_{i}|}\int_{a_{i}}^{b_{i} }f_{i}(\theta) d \theta$ equals the number of points separated from their closest center by a uniformly random threshold cut  $(i, \theta): \theta \in I_{i}$ along dimension $i$.  Thus the expected number of points separated from their closest center by a uniformly random threshold cut in $\allcuts$ is
 \[
 \sum_{i=1}^d \frac{|I_i|}{L} \cdot \frac{1}{|I_{i}|}\int_{a_{i}}^{b_{i} }f_{i}(\theta) d \theta = \frac{1}{L} \sum_{i=1}^d  \int_{a_{i}}^{b_{i} }f_{i}(\theta) d \theta \leq\frac{1}{L} \sum_{i=1}^d  \int_{-\infty}^{\infty }f_{i}(\theta) d \theta  = \cost_1(\mathcal{U})/L\,,
 \]
 where we used $f_i(\theta) \geq 0$ for the inequality.
\end{proof}

The above lemma upper bounds the expected number of points whose cost increases from a uniformly random threshold cut. We proceed to analyze how much this increase is, in expectation. Let $\treeleaves(t)$ denote the state of $\treeleaves$ at the beginning of the $t$-th iteration of the while loop of Algorithm~\ref{alg:random} and let $c_{\max}(t) = \max_{B\in \treeleaves(t)} \max_{\bm{\mu}^i, \bm{\mu}^j \in B} \| \bm{\mu}^i - \bm{\mu}^j\|_1$ denote the maximum distance between two centers that belong to the same leaf at the beginning of the $t$-th iteration. With this notation we have that $\treeleaves(1) =  \{\mathcal{U}\}$ and that $c_{\max}(1)$ equals the $c_{\max}$ in the statement of Theorem~\ref{thm:l1-alg}. Observe that $c_{\max}(t) \geq c_{\max}(t+1)$ and $c_{\max}(t)  = 0$ if $|\treeleaves| = k$ (i.e., when each leaf contains exactly one center). Understanding the rate at which $c_{\max}(t)$ decreases is crucial for our analysis because of the following observation: Consider a leaf $B \in \treeleaves(t)$ and a point $\bm{x} \in \X$ that has not yet been separated from its closest center $\pi(\bm{x}) \in B$. If the threshold cut selected in the $t$-th iteration separates $\bm{x}$ from $\pi(\bm{x})$ then the cost of $\bm{x}$ in the final threshold tree is upper bounded by 
$\max_{\bm{\mu} \in B} \|\bm{x} - \bm{\mu}\|_1$, which, by the triangle inequality, is at most 
\begin{equation}
\max_{\bm{\mu} \in B} \|\bm{x} - \pi(\bm{x})\|_1 + \|\pi(\bm{x}) - \mu\|_1 \leq \|\bm{x} - \pi(\bm{x})\|_1  + c_{\max}(t)\,.
\label{eq:l1_cost_increase}
\end{equation}
In other words, a point that is first separated from its closest center by the threshold cut selected in the $t$-th iteration has a cost increase of at most $c_{\max}(t)$. 
\begin{lemma}
    Fix the the threshold cuts selected by Algorithm~\ref{alg:random}  during the first $t-1$ iterations (this determines the random variable $\treeleaves(t)$ and thus $c_{\max}(t)$). Let $M = 3 \ln(k) \cdot {2L}/{c_{\max}(t)}$.  Then  
    \[
        \Pr[c_{\max}(t+M) \leq c_{\max}(t)/2]  \geq 1-1/k\,,
    \]
    where the probability is over the random cuts selected in iterations $t, t+1, \ldots, t+M-1$.
    \label{lemma:l1-drop}
\end{lemma}
\begin{proof}
    Consider two centers $\bm{\mu}^i$ and $\bm{\mu}^j$ that belong to the same leaf in $\treeleaves(t)$. The probability that a uniformly random threshold cut from $\allcuts$ separates these two centers equals $\|\bm{\mu}^i - \bm{\mu}^j\|_1/L$. Thus if the centers are at distance at least $c_{\max}(t)/2$, the probability that they are \emph{not} separated by any of $M$ independently chosen cuts is at most
    \[
    \left( 1- \frac{c_{\max}(t)/2}{L} \right)^M = \left( 1- \frac{c_{\max}(t)}{2L} \right)^{3 \ln(k) \cdot {2L}/{c_{\max}(t)}}  \leq (1/e)^{3 \ln(k)}  = 1/k^3\,.
    \]
    There are at most $\binom{k}{2}$ pairs of centers in the leaves of $\treeleaves(t)$ at distance at least $c_{\max}(t)/2$. By the union bound, we thus have, with probability at least $1-1/k$, that each of these pairs are separated by at least one of the cuts selected in iterations $t, t+1, \ldots, t+M-1$. In that case, any two centers in the same leaf of $\treeleaves(t+M)$ are at distance at most $c_{\max}(t)/2$ and so $c_{\max}(t+M) \leq c_{\max}(t)/2$.
\end{proof}

Equipped with the above lemmas we are ready to analyze the expected cost of the tree output by Algorithm~\ref{alg:random}.
  Let $(i_t, \theta_t)$ denote the cut selected by Algorithm~\ref{alg:random} in the $t$-th iteration. 
  As argued above in~\eqref{eq:l1_cost_increase}, $c_{\max}(t)$ upper bounds the cost increase of the points first separated from their closest center by the $t$-th threshold cut. Hence,
  \[
    \E\left[\cost_1(T)\right] \leq \cost_1(\mathcal{U}) + \E\left[\sum_t c_{\max}(t) f_{i_t}(\theta_t) \right],
  \]
  where the sum is over the iterations of Algorithm~\ref{alg:random} (and recall that $f_i(\theta)$ denotes the number of points separated from their closest center by the cut $(i,\theta)$). We remark that the right-hand side is an upper bound  (and not an exact formula of the cost) for two reasons:  first, not every separated point may experience a cost increase of $c_{\max}(t)$, and second, the right-hand side adds a cost increase every time a cut separates a point from its closest center and not only the first time. Nevertheless, we show that this upper bound yields the stated guarantee. We do so by analyzing the expected cost increase  of the cuts until $c_{\max}(t)$ has halved.   Specifically, let 
  \[\costinc(r) = \sum_{t\,:\, c_{\max}(t) \in (c_{\max}/2^{r+1}, c_{\max}/2^r]} c_{\max}(t) f_{i_t} (\theta_t)
  \]be the random variable that upper bounds the cost increase caused by the cuts selected during the iterations $t$ when $c_{\max}/2^{r+1} < c_{\max}(t) \leq c_{\max}/2^r$. Then \[\E[\cost_1(T)] \leq \cost_1(\mathcal{U}) + \sum_{r}  \E[\costinc(r)]\,,\] where the sum is over $r$ from $0$ to $1+ \lfloor \log_2 (c_{\max}/c_{\min})\rfloor$. 
  The bound on the expected cost therefore follows from the following lemma.
  \begin{restatable}[]{lemma}{costincrease}
      For every $r$,
      \(
    \E[\mbox{cost-increase}(r)] \leq 12\ln(k)\cdot \cost_1(\U). \)
  \label{lemma:cost-increase}
\end{restatable}
 Let $M = 3\ln(k) \cdot 2L/c_{\max}(t)$ as in Lemma~\ref{lemma:l1-drop}. Using Lemma~\ref{lemma:l1-nr-cut}, one can upper bound the expected cost of $M$ uniformly random cuts in iterations $t, t+1, \ldots, t + M-1$ by $6\ln(k) \cdot \cost_1(\U)$ and $M$ cuts is very likely to halve $c_{\max}(t)$ as in Lemma~\ref{lemma:l1-drop}.
 It is thus very likely that the cost of these $M$ cuts upper bounds $\costinc(r)$.
 The additional constant factor of $2$ in the statement of the lemma arises by considering the small ``failure'' probability of such a trial. 
 The formal proof bounding the expected cost of this geometric series can be found in
 \appref{appendix:cost-increase}.

\subsubsection{Upper bounding cost by a factor of \texorpdfstring{$O(\log^2 k)$}{O(log\string^2 k)}}
\label{sec:mod-l1-alg}
\label{sec:logk-upper-bound}

Observe that our analysis of Algorithm~\ref{alg:random}  implies that the expected cost of the output tree is at most $O(\log^2(k) \cdot \cost_1(\mathcal{U}))$ whenever $c_{\max}$ and $c_{\min}$ do not differ by more than a polynomial factor in $k$. 
However, our current techniques fail to upper bound this expectation by a factor $o(k)$ for general $c_{\max}$ and $c_{\min}$. 
To illustrate this point, consider the $k$-dimensional instance with a single point $\bm{x}$ at the origin and $k$ centers where the $i$-th center  $\bm{\mu}^i$ is located at the $i$-th standard basis vector scaled by the factor $2^i$. In our analysis, we upper bound the cost of  $\bm{x}$ with its \emph{maximum} distance to those  centers that remain in the same leaf whenever $\bm{x}$ is separated from its closest center $\bm{\mu}^1$. This yields the following upper bound on the expected cost of $\bm{x}$ in the final tree
\[
    \sum_{i=2}^k 2^i \Pr[\mbox{$\bm{x}$ is separated from $\bm{\mu}^1$ and $\bm{\mu}^i$ is the farthest remaining center}]\,.
\]
Due to the exponentially increasing distances, this is lower bounded by
\[\sum_{i=2}^k 2^{i-1} \Pr[\mbox{$\bm{x}$ is separated from $\bm{\mu}^1$ and $\bm{\mu}^i$ is a remaining center}].\]
Now note that the probability in this last sum equals 
\[
\Pr[\mbox{$\bm{x}$ is separated from $\bm{\mu}^1$ before $\bm{x}$ is separated from $\bm{\mu}^i$ }] = {2}/{(2 + 2^i)}.
\]
It follows that any analysis of Algorithm~\ref{alg:random} that simply upper bounds the reassignment cost of a point with the  maximum distance to a remaining center cannot do better than a factor of $\Omega(k)$.

We overcome this obstacle by analyzing a slight modification of Algorithm~\ref{alg:random} that avoids these problematic cuts that separate very close centers. 
Recall the notation used in the previous section: $\treeleaves(t)$ denotes the state of $\treeleaves$ at the beginning of the $t$-th iteration of the while loop and $c_{\max}(t) = \max_{B\in \treeleaves(t)} \max_{\bm{\mu}^i, \bm{\mu}^j \in B} \| \bm{\mu}^i - \bm{\mu}^j\|_1$ denotes the maximum distance between two centers that belong to the same leaf at the beginning of the $t$-th iteration. We now modify Algorithm~\ref{alg:random} by replacing Line~5 ``Sample \((i, \theta) \) uniformly at random from \(\allcuts\)'' by 
\begin{quote}
    Sample \((i, \theta) \) uniformly at random from those cuts in \(\allcuts\) that do \emph{not} separate two centers that are within distance at most $c_{\max}(t) /k^4$.
\end{quote}
This modification allows us to prove a nearly tight guarantee on the expected cost.
\begin{theorem}
    Given reference centers $\mathcal{U} = \{\bm{\mu}^1, \bm{\mu}^2, \ldots, \bm{\mu}^k\}$, \emph{modified} Algorithm~\ref{alg:random} outputs a threshold tree $T$ whose expected cost satisfies
    \(
      \E[\cost_1(T)] \leq O(\log^2 k ) \cdot \cost_1(\mathcal{U}).
    \)
    \label{thm:l1-modified-alg}
\end{theorem}

The above theorem implies the furthermore statement of Theorem~\ref{thm:l1-alg}. This follows by observing that  Algorithm~\ref{alg:random} selects the same cuts as the modified version with probability at least $1-1/k$. This error probability can be made smaller by not allowing cuts that separate centers within distance $c_{\max}(t)/k^\ell$ for an $\ell$ larger than $4$.
We give a more formal explanation of this implication in Appendix~\ref{sec:thm2-implies-thm1}.

The proof of Theorem~\ref{thm:l1-modified-alg} is similar to the cost analysis in Section~\ref{sec:l1-cost-analysis} with the main difference being that here we are more careful in bounding the cost when considering different ``rounds'' of the algorithm. In the analysis it will be convenient to take the following viewpoint of the modified algorithm: it samples a uniformly random cut and then discards it if it separates two centers within distance $c_{\max}(t)/k^4$. While the number of iterations may increase with this viewpoint, the output distribution is the same as the modified algorithm in that, in each iteration, a cut is sampled uniformly at random among those that do not separate any centers within distance $c_{\max}(t)/k^4$.  In the following, we refer to this as the \emph{sample-discard algorithm} and we prove Theorem~\ref{thm:l1-modified-alg} by showing that the sample-discard algorithm outputs a tree whose expected cost is $O(\log^2 k) \cdot \cost_1(\mathcal{U})$.

Let $(i_t, \theta_t)$ denote the (uniformly random) cut selected in the $t$-th iteration of the sample-discard algorithm and recall the following notation: $\treeleaves(t)$ denotes the state of $\treeleaves$ at the beginning of the  $t$-th iteration of the while-loop and $c_{\max}(t) = \max_{B\in \treeleaves(t)} \max_{\bm{\mu}^i, \bm{\mu}^j \in B} \| \bm{\mu}^i - \bm{\mu}^j\|_1$ denotes the maximum distance between two centers that belong to the same leaf at the beginning of the $t$-th iteration. 
We start by observing that  Lemma~\ref{lemma:l1-drop} readily generalizes to the modified version.
\begin{lemma}
Fix the the threshold cuts selected by the sample-discard algorithm  during the first $t-1$ iterations (this determines the random variable $\treeleaves(t)$ and thus $c_{\max}(t)$). Let $M = 3 \ln(k) \cdot {4L}/{c_{\max}(t)}$.  Then  
\[
    \Pr[c_{\max}(t+M) \leq c_{\max}(t)/2]  \geq 1-1/k\,,
\]
where the probability is over the random cuts selected in iterations $t, t+1, \ldots, t+M-1$.
\label{lemma:logk-l1-drop}
\end{lemma}
\begin{proof}
    The proof is similar to that of Lemma~\ref{lemma:l1-drop} but some care has to be taken  as certain cuts are now discarded.
    
    Consider two centers $\bm{\mu}^i$ and $\bm{\mu}^j$ that belong to the same leaf in $\treeleaves(t)$. Further suppose that $c_{\max}(t)/2 \leq \| \bm{\mu}^p - \bm{\mu}^q \|_1 \leq c_{\max}(t)$. We have that  any cut $(i,\theta)$ that separates these two centers is considered  (i.e., not discarded) by the sample-discard algorithm after the $t$-th iteration unless $(i,\theta)$ also separates two centers within distance $c_{\max}(t)/k^4$. Here we used that $c_{\max}(\cdot)$ is monotonically decreasing and so the set of cuts that are discarded if sampled is only decreasing in later iterations.      We can thus obtain the lower bound $c_{\max}(t)/(4L)$ on  the probability that a uniformly random cut  separates $\bm{\mu}^p$ and $\bm{\mu}^q$ by subtracting 
   \[
   \frac{1}{L}\sum_{i=1}^d \int_{-\infty}^{\infty} \mathbbm{1}\left[\mbox{$\theta$ separates two centers within distance } {c_{\max}(t)}/{k^4}\right]\, d\theta 
   \leq\frac{1}{L} \binom{k}{2} \frac{c_{\max}(t)}{{k^4}}
   \]
   from 
   \[
    \frac{1}{L}\sum_{i=1}^d \int_{-\infty}^{\infty} \mathbbm{1}[\mbox{$\theta$ between $\bm{\mu}_i^p$ and $\bm{\mu}_i^q$}]\, d\theta \geq \frac{1}{L} c_{\max}(t)/2\,.
   \]
    The proof now proceeds in the exact same way as that of Lemma~\ref{lemma:l1-drop}.
    Indeed, if the centers are at distance at least $c_{\max}(t)/2$, the probability that they are \emph{not} separated by any of $M$ independently chosen cuts is at most
    \[
    \left( 1- \frac{c_{\max}(t)}{4L} \right)^M = \left( 1- \frac{c_{\max}(t)}{4L} \right)^{3 \ln(k) \cdot {4L}/{c_{\max}(t)}}  \leq (1/e)^{3 \ln(k)}  = 1/k^3\,.
    \]
    There are at most $\binom{k}{2}$ pairs of centers in the leaves of $\treeleaves(t)$ at distance at least $c_{\max}(t)/2$. By the union bound, we thus have, with probability at least $1-1/k$, that each of these pairs are separated by at least one of the cuts selected in iterations $t, t+1, \ldots, t+M-1$. In that case, any two centers in the same leaf of $\treeleaves(t+M)$ are at distance at most $c_{\max}(t)/2$ and so $c_{\max}(t+M) \leq c_{\max}(t)/2$.
\end{proof}

Now, similar to Section~\ref{sec:l1-cost-analysis}, we can upper bound the expected cost of the constructed threshold tree $T$ by 
\[
    \E\left[\cost_1(T)\right] \leq \cost_1(\mathcal{U}) + \E\left[\sum_t c_{\max}(t) f_{i_t}(\theta_t) \mathbbm{1}[\mbox{$(i_t,\theta_t)$ was added to the tree}] \right]\,,
  \]
  where the sum is over the iterations. We remark that, in contrast to Section~\ref{sec:l1-cost-analysis}, we have strengthened the upper bound by only considering those cuts that were actually added to the threshold tree by the modified algorithm. This refinement is necessary for obtaining the improved guarantee. We now analyze the sum in the expectation by partitioning it into $1+ \lfloor \log_2 (c_{\max}/c_{\min})\rfloor$ rounds. Specifically for $r\in \{0,1 \ldots, \lfloor \log_2(c_{\max}/c_{\min})\rfloor \}$, we let
  \[
    \costinc'(r) = \sum_{t: c_{\max}(t) \in (c_{\max}/2^{r+1}, c_{\max}/2^r]} c_{\max}(t) f_{i_t} (\theta_t) \mathbbm{1}[\mbox{$(i_t, \theta_t)$ was added to the tree}] 
  \] be the cost of the cuts selected during the iterations $t$ when $c_{\max}/2^{r+1} < c_{\max}(t) \leq c_{\max}/2^r$.
  
  To upper bound $\E[\costinc'(r)]$ we use $\activecut_r(i,\theta)\in \{0,1\}$ to denote the indicator variable of those cuts that separate two centers within distance $c_{\max}/2^r$ and do not separate any two centers within distance $c_{\max}/(2^{r+1}k^4)$.
\begin{lemma}
    For a round $r$, 
    \[
        \E[\costinc'(r)] \leq 
    24 \ln(k) \cdot   \sum_{i=1}^d \int_{-\infty}^\infty f_i(\theta) \, \activecut_r(i,\theta) \, d\theta\,.
    \]
    \label{lemma:log-l1-drop}
\end{lemma}

Before giving the proof of this lemma, let us see how it implies Theorem~\ref{thm:l1-modified-alg}. For this, note that a cut $(i,\theta)$ only has $\activecut_r(i,\theta) = 1$ for at most $O(\log(k^4))$ many values of $r$. Indeed, let $c$ be the distance between the closest centers that $(i,\theta)$ separates. Then any round $r$ for which $\activecut_r(i,\theta) = 1$ must satisfy $c_{max}/(2^{r+1} k^4) \leq c \leq c_{max}/2^r$. Hence, we have 
\begin{align*}
    \cost_1(T) &\leq \cost_1(\mathcal{U}) + \sum_r \E[\costinc'(r)]  \\
    & \leq \cost_1(\mathcal{U}) + \sum_r  24 \ln(k) \cdot   \sum_{i=1}^d \int_{-\infty}^\infty f_i(\theta) \, \activecut_r(i,\theta) \, d\theta \\
    & \leq \cost_1(\mathcal{U}) + O(\log^2 k) \cdot   \sum_{i=1}^d \int_{-\infty}^\infty f_i(\theta)   \, d\theta  \\
    & = O(\log^2 k) \cdot \cost_1(\mathcal{U})\,.
\end{align*}
In other words, we proved that the sample-discard algorithm outputs a tree $T$ with $\E[\cost_1(T)] \leq O(\log^2 k) \cdot \cost_1(\mathcal{U})$, which implies Theorem~\ref{thm:l1-modified-alg} since modifed Algorithm~\ref{alg:random} and the sample-discard algorithm have the same output distribution.
It remains to prove the lemma.
\begin{proof}[Proof of Lemma~\ref{lemma:log-l1-drop}]
    Consider the first iteration $t$ such that $c_{\max}(t) \leq c_{\max}/2^r$. Further suppose that $c_{\max}(t) > c_{\max}/2^{r+1}$ since otherwise $\costinc'(r) = 0$ and the statement is trivial.  We proceed to upper bound $\E[\costinc'(r)]$ as follows. First note that the cost of a random cut sampled in an iteration $t'$ such that $c_{\max}/2^{r+1} \leq c_{\max}(t') \leq c_{\max}(t)$ equals
    \[
        \frac{c_{\max}(t')}{L}\sum_{i=1}^d \int_{-\infty}^\infty f_i(\theta)  \mathbbm{1}[\mbox{$(i,\theta)$ was added to the tree}] \, d\theta\,.
    \]
    The cut $(i,\theta)$ can be added to the tree only if it does not separate any centers within distance $c_{\max}(t')/k^4 \geq c_{\max}/(2^{r+1} k^4)$ and it must separate two centers within distance at most $c_{\max}(t') \leq c_{\max}/2^r$. In other words, any cut that is added to the tree must have $\activecut_r(i,\theta) = 1$. We can thus upper bound the above cost of a single cut by
    \begin{equation}
        \frac{c_{\max}(t)}{L}\sum_{i=1}^d \int_{-\infty}^\infty f_i(\theta) \, \activecut_r(i,\theta) \, d\theta\,,
        \label{eq:logk-one-cut}
    \end{equation}
    where we also used that $c_{\max}(t') \leq c_{\max}(t)$.

   The proof now follows arguments that are again similar to those in Section~\ref{sec:l1-cost-analysis}. Select $M = 12\ln(k)\cdot L/c_{\max}(t)$ as in Lemma~\ref{lemma:logk-l1-drop}. We upper bound $\E[\costinc(r)]$ by adding ``trials'' of $M$ cuts until $c_{\max}(\cdot)$ goes below $c_{\max}/2^{r+1}$. (Strictly speaking this may not happen after a multiple of $M$ cuts but considering more cuts may only increase the cost of our upper bound.)  Let $H$ be the event that the following $M$ cuts causes $c_{\max}(\cdot)$ to drop below $c_{\max}/2^{r+1}$. By Lemma~\ref{lemma:logk-l1-drop}, $\Pr[H] \geq 1-1/k$. Furthemore, the expected cost of $M$ cuts is $M$ times~\eqref{eq:logk-one-cut} which equals
   \[
    12\ln(k) \cdot   \sum_{i=1}^d \int_{-\infty}^\infty f_i(\theta) \,  \activecut_r(i,\theta) \, d\theta\,.
   \]
   
   The statement now follows from the same ``geometric distribution'' calculations as in the proof of Lemma~\ref{lemma:cost-increase}. 
\end{proof}

\subsection{Implementation details}
\label{sec:implementation}

Since we only use cuts that split at least one leaf, the algorithm in fact only needs to sample cuts conditioned on this event. 
Note that if we sample only the cuts that split at least one leaf, the while loop in \cref{line:while} of~\cref{alg:random} runs for at most \(k - 1 \) iterations.
We now explain how to efficiently sample cuts (Line~\ref{line:sample}), find the leaves split by a given cut (Line~\ref{line:for}), and implement the split operation (Lines~\ref{line:split}--\ref{line:update_leaves}).

We first show how to efficiently implement the split operation.
For a cut \( (i, \theta) \) that splits a given leaf \( B \) into \(B^-\) and \(B^+\), the split operation can be implemented in \(O(d \cdot \min(|B^-|,|B^+|) \cdot \log |B|) \) time as follows: 
In each leaf \( B \), we maintain \(d\) binary search trees \(T^B_1, \dots, T^B_d\) where \(T^B_i\) stores the \(i\)-th coordinate of the centers in \(B\). 
Now, given a cut \((i, \theta)\) and a leaf \( B \) that gets split by \( (i, \theta) \), we can find the number of centers in \(B\) that have a smaller or equal \(i\)-th coordinate than \(\theta\) using \(T^B_i\) in \(O(\log |B|) \) time. 
Let \(b^- \) be this number, and let \(b^+ = |B|-b^- \). 
Suppose that \( b^- \leq b^+ \). For the other case, the implementation is analogous. 
In this case, we construct \(B^-\) by initializing it with \(d\) empty binary search trees and inserting the centers whose \(i\)-th coordinate is at most \(\theta\) to each of them. This takes \(O(d \cdot b^- \cdot \log |B|) \) time. For \(B^+\), we just reuse the binary search trees of \(B\) after removing the centers that belong to \(B^-\). This also takes \(O(d \cdot b^- \cdot \log |B|) \) time.
Let \(\tau(k)\) denote the running time of all splitting operations performed by the algorithm when starting with a single leaf with \(k\) leaves. Then, \(\tau(k) = \tau(k-k') + \tau(k') + O(d \cdot \min(k-k',k') \cdot \log k) \) and by induction, we conclude that \(\tau(k) = O(d k \log^2 k) \).

To find the leaves that get separated by a cut, we employ the following data structure.
For each dimension \(i\), we maintain a balanced interval tree \(T^{\operatorname{int}}_i\). 
For each tentative leaf node with centers \(B\), we store the interval indicating the range of the \(i\)-th coordinate of \(B\) in \(T^{\operatorname{int}}_i\).
We update the corresponding interval trees after each split operation, which amounts to removing at most one interval and adding at most two intervals per node that gets split.
Note that the added and removed intervals for a single split operation for a fixed dimension can be computed in \(O(\log k)\) time using the previously described node binary search trees.
Moreover, adding and removing intervals to and from an interval tree with at most \(k\) intervals also takes \(O(\log k)\) time.
As we have  \(O(k)\) split operations in total, the time to maintain the interval trees is \(O(d k \log k)\).
Now, given a cut \((i,\theta)\), we can retrieve all the leaves that get separated by \((i,\theta)\) in \(O(\log (k_1) + k_2) \) time where \(k_1\) is the number of tentative leaves and \(k_2\) is the number of tentative leaves that get separated by the cut. 
To retrieve such leaves, we query the \(i\)-th interval tree to find all intervals that contain the value \(\theta\). 
Since we sample at most \(k-1\) cuts from the conditioned distribution, the total time for this operation over all cuts and all dimensions is \(O(d k \log k) \).

What remains is to show that we can efficiently sample a uniform cut conditioned on the event that it splits at least one leaf. To this end, in the interval trees described above, we also maintain the lengths of the unions of intervals in each subtree. This length information can be updated in \(O(\log k')\) time where \(k' \leq k \) is the number of intervals in an interval tree. 
Then in \(O(d)\) time, one can sample a dimension \(i\) and in \(O(\log k')\) time, sample a suitable \(\theta\) value.

\section{Explainable \texorpdfstring{$k$}{k}-means and general \texorpdfstring{$\ell_p$}{lp}-norm clustering}
\label{ell-p}

\newcommand{\closecenter}{\operatorname{\pi}}
\newcommand{\propdist}{\mathcal{D}_p}
\newcommand{\intervals}{\mathcal{I}}
\newcommand{\seg}{\operatorname{I}}

In this section, we generalize \cref{thm:l1-modified-alg} to the explainable \(k\)-clustering problems with assignment cost defined in terms of the \(\ell_p\)-norm, which includes the explainable \(k\)-means (\(p = 2\)) problem.

Recall that in \cref{ell-1}, we sample cuts from the uniform distribution over \(\allcuts\), and consequently, the probability that a point \(\bm{x} \in \mathcal{X} \) is separated from its closest center \(\closecenter(\bm{x})\) is proportional to the \(\ell_1\) distance between \(\bm x\) and \(\closecenter(\bm{x})\).
However, selecting cuts according to the uniform distribution can be arbitrarily bad for higher \(p\)-norms even in one-dimensional space.
For example, consider the \(k\)-means (i.e. \(p = 2\)) problem with \(d = 1\) where the cost of assigning a point \(\bm{x}\) to a center \(\bm{\mu}\) is defined as \( \|\bm{x} -\bm{\mu}\|^2_2 \). 
Suppose we have two centers \(\bm{\mu}^1 = -1\) and \(\bm{\mu}^2 = D > 1\), and fix one data point \(\bm x = \bm{0}\). 
The closest center to \(\bm x\) is \(\bm{\mu}^1\) and hence the original cost is \(1\). However, the expected cost of a uniformly random cut is \(((D\cdot 1^2 + 1 \cdot D^2)/(1+D)  = D \) which can be arbitrarily large. 

To avoid such drastic costs, we sample cuts from a generalized distribution.
Ideally, we would like to sample cuts analogously to the case of \(k\)-medians so that the probability that we separate a point \(\bm x\) from its closest center \(\closecenter(\bm x)\) is proportional to \( \|\bm x - \closecenter(\bm x) \|^p_p \).
However, sampling from such a distribution seems very complicated if at all possible.
Instead, we sample from a slightly different distribution:
Namely, for a \(p\)-norm where the cost of assigning a point \(x\) to a center \(y\) is \( \|x-y\|^{p}_{p}\), we sample cuts \((i,\theta)\) from the distribution where the probability density function of \( (i,\theta) \) is proportional to \( \min_{j \in [k]} |\mu^j_i - \theta|^{p-1} \), the \((p-1)\)-th power of the minimum distance to a center \emph{along the \(i\)-th dimension}. We call this distribution \(\propdist \).

Using samples from \(\propdist\) with a modified version of~\cref{alg:random} yields~\cref{thm:lp-modified-alg}.
\begin{restatable}[]{theorem}{lpmodifidalgthm}
  For every \(p\geq1\), there exists a randomized algorithm that when given input centers $\mathcal{U} = \{\bm{\mu}^1, \bm{\mu}^2, \ldots, \bm{\mu}^k\}$, outputs a threshold tree $T$ whose expected cost satisfies
    \[
      \E[\cost_p(T)] \leq O(k^{p-1} \log^2 k) \cdot \cost_p(\mathcal{U}).
    \]
    \label{thm:lp-modified-alg}
\end{restatable}

\label{sec:proof-lp-norm}

To prove \cref{thm:lp-modified-alg}, we start by introducing some notation and making the definition of \(\propdist\) precise.
For a dimension \(i \in [d]\), let \(\mu_i^- = \min_{j \in [k] }\mu^j_i \) and \(\mu_i^+ =  \max_{j \in [k]}\mu^j_i \). 
For a dimension \(i \in [d]\) and two coordinates \(x, y \in \R\), let \(\intervals_i(x, y) \) be the set of consecutive intervals along the \(i\)-th dimension delimited by the coordinates \(x\) and \(y\) themselves and the projections of the centers in \( \mathcal{U} \) that lie between \(x\) and~\(y\). 
For example, consider the \(2\)-dimensional instance with four centers \(\bm\mu^1, \dots, \bm \mu^4\) shown in \cref{fig:intervals-on-dim}. 
On the horizontal axis, two coordinates \(x\) and \(y\) are marked along with the projections of the four centers \(\mu_1^1, \mu_1^2, \mu_1^3\), and \(\mu_1^4\).
Here, \(\intervals_1(x,y) \) consists of the three consecutive intervals \( [x, \mu^4_1], [\mu^4_1, \mu^2_1]\), and~\([\mu^2_1, y]\). 

\begin{figure}[ht]
\centering
\begin{tikzpicture}[scale=1.5]
    \draw [<->,thick] (0,2) node (yaxis) [above] {2} |- (3,0) node (xaxis) [right] {1};
  
    \coordinate (x) at (0.8, 0);
    \coordinate (y) at (2.3, 0);
    \coordinate (a) at (0.5, 1.2);
    \coordinate (b) at (1.3, 1.8);
    \coordinate (c) at (1.8, 0.6);
    \coordinate (d) at (2.8, 0.8);
    
    \draw[dashed] (xaxis -| x) node[below] {$x$};
    \draw[dashed] (xaxis -| y) node[below] {$y$};
    
    \draw[dashed] (xaxis -| b) node[below] {$\mu^4_1$};
    \draw[dashed] (b) -- (1.3,0);
    
    \draw[dashed] (xaxis -| c) node[below] {$\mu^2_1$};
    \draw[dashed] (c) -- (1.8,0);
    
    \draw[dashed] (xaxis -| a) node[below] {$\mu^1_1$};
    \draw[dashed] (a) -- (0.5,0);
    
    \draw[dashed] (xaxis -| d) node[below] {$\mu^3_1$};
    \draw[dashed] (d) -- (2.8,0);
    
    \fill[black] (a) circle (1pt) ;     
    \fill[black] (b) circle (1pt) ;     
    \fill[black] (c) circle (1pt) ;     
    \fill[black] (d) circle (1pt) ;     
    \draw [fill=white] (x) circle (1pt) ;     
    \draw [fill=white] (y) circle (1pt) ;     
     
    \draw (a) node [above]{\(\bm{\mu}^1\)};
    \draw (b) node [above]{\(\bm{\mu}^4\)};
    \draw (c) node [above]{\(\bm{\mu}^2\)};
    \draw (d) node [above]{\(\bm{\mu}^3\)};
    \draw (x) node [above]{};
    \draw (y) node [above]{};
\end{tikzpicture}
    \caption{Intervals defined by projecting points onto a coordinate axis.}
    \label{fig:intervals-on-dim}
\end{figure}
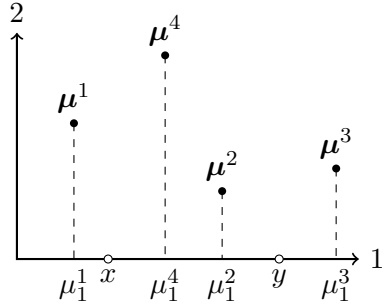

Observe that, by the definition of \(\intervals_i(x, y) \), we have \(|x - y| = \sum_{[a,b] \in \intervals_i(x, y)} |b - a| \).

\newcommand\all{\operatorname{all}}
Let \[\intervals_{\all} = \bigcup_{i\in [d]} \left\{ (i, [a,b]) : [a,b] \in \intervals_i(\mu_i^-, \mu_i^+) \right\} \]
denote the collection of all dimension--interval pairs that are delimited by the projections of the centers onto the respective dimensions.
We define \[L_p = \sum_{(i,[a,b]) \in \intervals_{\all}} |b-a|^p. \]

With the introduced notation, the distribution \( \propdist \) can be formally described as follows: We first select a dimension \(i\) and an interval \( [a, b] \in \intervals_i(\mu_i^-, \mu_i^+) \) along with dimension \(i\)  (i.e., we select a dimension--interval pair \((i,[a,b]) \in \intervals_{\all}\)) with probability  \( |b-a|^p/L_p \).
Then we pick \(\theta \in [a, b] \) randomly such that the p.d.f. \( \theta \) is \[P_{a,b}(\theta) := \frac{p\cdot 2^{p-1}}{(b-a)^p} \min(\theta - a, b - \theta)^{p-1}. \]

Another key component of the design and analysis of the generalized algorithm is a \emph{pseudo-distance} function.
For two points \(\bm x, \bm y \in \R^d \), Let \[\intervals(\bm x, \bm y) = \bigcup_{i\in[d]} \{ (i, [a,b]) : [a,b] \in \intervals_i(x_i,y_i) \}. \] 
We then define the pseudo-distance between \(\bm x\) and \(\bm y\) as 
\[d_p(\bm x, \bm y) = \sum_{(i, [a,b]) \in \intervals(\bm x, \bm y)} |b-a|^p.\]
Note that the \(p\)-th power of the \(\ell_p\) distance, \( \| \bm x - \bm y \|^p_p \), between two points \(\bm x\) and \(\bm y\) is defined as \( \sum_{i \in [d]} |x_i - y_i|^p \).
It is easy to see that \(\|\bm x-\bm y\|^p_p \geq d_p(\bm x, \bm y) \) since
\[
|x_i - y_i|^p = \left( \sum_{[a,b] \in \intervals_i(x_i, y_i)} |a-b|\right)^p \geq  \sum_{[a,b] \in \intervals_i(x_i, y_i)} |a-b|^p
\] for each dimension \(i\). For \(p = 1\), we have equality.

A key observation now is that, if we sample a cut from \(\mathcal{D}_p\), the probability that it separates two centers \(\bm\mu^g\) and \(\bm\mu^h \) is \emph{proportional to their pseudo-distance} \(d_p(\bm\mu^g,\bm\mu^h)\).

\subsection{The algorithm for \texorpdfstring{$\ell_p$}{lp}-norms with \texorpdfstring{$p \geq 1$}{p >= 1}.}

We now present the generalized algorithm. 
The only difference from the modified version of \cref{alg:random} is how we sample random cuts at~\cref{line:sample}.
Recall from \cref{ell-1} that we defined \(\treeleaves(t)\) to denote the state of \(\treeleaves\) at the beginning of the \(t\)-th iteration.
We define \(c'_{p, \max}(t)\) as the maximum pseudo-distance between any pair of centers in a leaf in \(\treeleaves(t)\).
Formally, \(c'_{p,\max}(t) = \max_{B \in \treeleaves(t)} \max_{\bm{\mu}^i,\bm{\mu}^j \in B} d_p(\bm{\mu}^i,\bm{\mu}^j)\).
Let \(c'_{p,\max} = c'_{p,\max}(1)\).

Now, in the sampling step (\cref{line:samplep}), we draw samples from \(\mathcal{D}_p\). 
However, we discard the cut if it separates any two centers in a leaf whose pseudo-distance is at most \(c'_{p,\max}(t)/k^4\).
Note that this is a generalization of the sample-discard algorithm from~\cref{sec:logk-upper-bound}.
We present the pseudo-code in~\cref{alg:randomp}.

\begin{algorithm}[ht]
  \DontPrintSemicolon
  \textbf{Input: } A collection of \( k \) centers \(\mathcal{U} = \{\bm{\mu}^1, \bm{\mu}^2, \ldots, \bm{\mu}^k\} \subset \R^d \). \;
  \textbf{Output: } A threshold tree with \(k\)-leaves. \;
  \(\treeleaves \gets \{ \mathcal{U} \}\) \; 
  \While{\( |\treeleaves| < k \) } { \label{line:whilep}
    Sample a cut \(\left( i,\theta \right) \) from \(\mathcal{D}_{p}\) \label{line:samplep} \; 
    \If{\((i,\theta)\) separates two centers that are closer than \(c'_{p,\max}(\cdot)/k^4\) in pseudo-distance}{ Discard the cut. \label{line:discard}}
    \Else{
    \For{each \(B \in \treeleaves \) that are split by \( (i, \theta) \)} { \label{line:forp}
      Split \(B\) into \(B^-\) and \(B^+\) and add them as left and right children of \(B\).\label{line:splitp} \; 
      Update \(\treeleaves\). \label{line:update_leavesp}
    }
    }
  }
  \textbf{return} the threshold tree defined by all cuts that separated some \( B \). \;
  \caption{Generalized explainable clustering algorithm for higher \(\ell_p\)-norms.}   
  \label{alg:randomp}
\end{algorithm}

Following the lines of \cref{sec:logk-upper-bound}, we now upper bound the expected cost of \cref{alg:randomp}. 

\begin{lemma}
    Fix the threshold cuts selected by \cref{alg:randomp} during the first \(t-1\) iterations. 
    Let \(M = 3 \cdot 4 \cdot \ln(k) \cdot L_p/c'_{p,max}(t)\). Then
    \[\Pr[c'_{p,\max}(t+M) \leq c'_{p,\max}(t)/2] \geq 1 - 1/k, \]
    where the probability is over the random cuts selected in iterations \(t, t+1, \dots, t+M-1\).
    \label{lemma:logk-lp-drop}
\end{lemma}

\begin{proof}
    \newcommand{\badpairs}{\operatorname{TooClose}}
    \newcommand{\badintervals}{\intervals_{\operatorname{bad}}}
    
    We begin by introducing a few more notations that are useful in the analysis.
    For an iteration~\(t\), let \(\badpairs(t)\) be the set of pairs of centers \((\bm{\mu}^g, \bm{\mu}^h)\) that satisfy \(d_p(\bm{\mu}^g,\bm{\mu}^h) \leq c'_{p,\max}(t)/k^4\). 
    In other words, \(\badpairs(t)\) contains pairs of centers that the algorithm is not allowed to separate at the \(t\)-th iteration.
    Note that for any \((\bm{\mu}^g, \bm{\mu}^h) \in \badpairs(t)\), both \(\bm{\mu}^g\) and \( \bm{\mu}^h\) will be in the same leaf in \(\treeleaves(t)\).
    Let \[\badintervals(t) = \bigcup_{(\bm\mu^g,\bm\mu^h) \in \badpairs(t)} \intervals(\bm\mu^g, \bm\mu^h) \] be the set of dimension--interval pairs \((i,[a,b])\) such that making a cut in interval \([a,b]\) along dimension \(i\) will separate a pair of centers in \(\badpairs\).
    Observe that a cut that is made outside of \(\badintervals(t)\) will not separate any pair of centers in \(\badpairs\).
    
    Consider a leaf \(B \in \treeleaves(t) \) and two centers \(\bm{\mu}^g\) and \(\bm{\mu}^h\) in \( B \) such that \(c'_{p,\max}(t)/2 \leq d_p(\bm{\mu}^g, \bm{\mu}^h) \leq c'_{p,\max}(t)\).
      
    Note that
    \begin{align*}
    \sum_{[a,b]\in\badintervals(t)} |b-a|^p  &\leq \sum_{(\bm\mu^{g'}, \bm\mu^{h'}) \in \badpairs(t)} \sum_{(i, [a,b]) \in \intervals(\bm\mu^{g'}, \bm\mu^{h'})} |b-a|^p \\
    &= \sum_{(i, [a,b]) \in \intervals(\bm\mu^{g'}, \bm\mu^{h'})} d_p(\bm\mu^{g'}, \bm\mu^{h'})  \\
    &\leq \binom{k}{2} \frac{c'_{p,\max}(t)}{k^4} \leq \frac{c'_{p,\max}(t)}{4}.
    \end{align*}
    In the last inequality, we use that \(k \geq 2\).
      
    Hence, the probability that a cut selected at the \(t\)-th iteration separates \(\bm\mu^g\) and \(\bm\mu^h\) is at least
    \[ \frac{d_p(\bm\mu^g, \bm\mu^h)}{L_p} -  \frac{\sum_{[a,b]\in\badintervals(t)} |b-a|^p}{L_p} \geq \frac{c'_{p,\max}(t)}{2L_p} - \frac{c'_{p,\max}(t)}{4L_p} \geq \frac{c'_{p,\max}(t)}{4L_p}. \] 
    The proof now follows by replacing \(c_{\max}(t)\) with \(c'_{p,\max}(t)\) and \(L\) with \(L_p\) in the remaining part of the proof of \cref{lemma:logk-l1-drop}.
\end{proof}

In the following analysis, we use the H\"{o}lder's inequality stated below:
\begin{lemma}[H\"{o}lder's inequality]
   For two real numbers \(u\) and \(v\) such that \(1/u + 1/v = 1\) and two positive real number sequences \(y_1, \dots, y_m\) and \(z_1, \dots, z_m\), it holds that 
   \[ \sum_{i\in[m]} y_i z_i \leq \left( \sum_{i\in[m]}y_i^u \right)^{1/u} \left( \sum_{i\in[m]}z_i^v \right)^{1/v}.\]
   In particular, setting \(y_1 = y_2 = \dots = y_m = 1\), \(u=p/(p-1)\) and \(v=p\) for some \(p\), and taking the \(p\)-th power on both sides, it holds that
   \[ \left(\sum_{i\in[m]}z_i\right)^p \leq m^{p-1} \sum_{i\in[m]} z_i^p .\] 
\end{lemma}

We now upper bound the expected cost. 
Recall that \(\closecenter(\bm x)\) denotes the closest center in \(\mathcal{U}\) to a point \(\bm x \in \bm{\mathcal{X}}\) and that  \(\cost_p(\mathcal{U})\) is defined as
 \[\cost_p(\mathcal{U}) = \sum_{\bm x \in \bm{\mathcal{X}}} \|\bm x - \closecenter(\bm x)\|^p_p = \sum_{\bm x \in \bm{\mathcal{X}}} \sum_{i \in [d]} |x_i - \closecenter(\bm x)_i |^p. \]
To bound the cost of the output clustering in the \(k\)-medians setting, we used the triangle inequality.
For general \(p\)-th power of \(p\)-norms, we use the following generalized triangle inequality:
\begin{lemma}
  Consider three points \(\bm x, \bm y, \bm z \in \R^d\). We have
  \(\|\bm z - \bm x\|^p_p \leq 2^{p-1}\left(\|\bm z - \bm y\|^p_p+ \| \bm y - \bm x\|^p_p\right) \). \label{lemma:trianglep}
\end{lemma}
\begin{proof}
Expanding \(\| \cdot \|^p_p\) as a summation over \(d\) dimensions, it is sufficient to prove that for any three real numbers \(x, y, z \in \R \), \(|z-x|^p \leq 2^{p-1}(|z-y|^p + |y-x|^p)\). 
    Without loss of generality, assume that \(z \geq x\). If \(y \leq x\) or \(y \geq z\), the proof follows trivially because we have \(|z-x| \leq |z - y|\) or \( |z-x| \leq |y-z| \), respectively.
    Now suppose that \(x \leq y \leq z\). 
    Let \(a=y-x\) and \(b=z-y\). 
    Since \(a+b = z-x\), we simply need to prove that \((a+b)^p \leq 2^{p-1}(a^p + b^p)\) which follows from H\"{o}lder's inequality.
\end{proof}

Recall that we defined \(c'_{p,\max}(t)\) and \(c'_{p,\max}\) earlier using the pseudo-distance function \(d_p\). 
We now define \(c_{p,\max}(t)\) and \(c_{p,\max}\) similarly, but using the \(p\)-th power of the \(\ell_p\) norm: Namely,
\(c_{p,\max}(t) = \max_{B \in \treeleaves(t)} \max_{\bm{\mu}^i,\bm{\mu}^j \in B} \|\bm{\mu}^i - \bm{\mu}^j\|^p_p \) and \( c_{p,\max} = c_{p,\max}(1) \).
We again use \((i_t, \theta_t)\) to denote the cut selected by \cref{alg:randomp} in the \(t\)-th iteration and \(f_i(\theta)\) to denote the number of points \(\bm x \in \mathcal{X}\) that are separated from \( \closecenter(\bm x) \) by a cut \((i,\theta)\).
 
For a point \(\bm x \in \mathcal{X}\), suppose that it is assigned to some center \(\bm \mu\) in the final threshold tree. If \(\bm \mu = \closecenter(\bm x)\), the cost contribution of \(\bm x \) in the final clustering is the same as that in the original clustering.
Suppose \(\bm \mu \neq \closecenter(\bm x)\) and suppose that \(\bm x\) was separated from \(\closecenter(\bm x)\) at iteration \(t\).
Then, using \cref{lemma:trianglep}, we conclude that the cost of assigning \(\bm x\) to \(\bm\mu\), i.e., \(\|\bm x - \bm \mu\|^p_p\), is upper bounded by
\[ 
2^{p-1} \left( \|\bm x - \closecenter(\bm x)\|^p_p + \| \closecenter(\bm x) - \mu \|^p_p \right) \leq 2^{p-1} \left( \|\bm x - \closecenter(\bm x)\|^p_p + c_{p,\max}(t) \right).
\]
\newcommand{\usedcuts}{\operatorname{UsedCuts}}
Let \(\usedcuts\) be the set of cuts used to split some leaf in \cref{line:splitp} of \cref{alg:randomp}.
Now using the above observation, we can upper bound the expected cost of the output tree, \(\E[ \cost_p(T)]\), by
\[\cost_p(T) \leq 2^{p-1}\left( \cost_p(\mathcal{U}) +  \sum_{r = 0}^{\infty} \E[\costinc(r)] \right) \]
where 
\[
\costinc(r) = \sum_{t: \frac{c'_{p,\max}}{2^{r+1}}\leq c'_{p,\max}(t) \leq \frac{c'_{p,\max}}{2^{r}}} c_{p,\max}(t) \cdot f_{i_t}(\theta_t)  \cdot \mathbbm{1}[(i_t,\theta_t)  \in \usedcuts].
\]
Note that in the last expression, the summed terms use \(c_{p,\max}(t)\) whereas the condition of the summation uses \(c'_{p, \max}(t)\).
Note that
\begin{align*}
   \costinc(r) &\leq  \sum_{t: \frac{c'_{p,\max}}{2^{r+1}}\leq c'_{p,\max}(t) \leq  \frac{c'_{p,\max}}{2^{r}}} k^{p-1}c'_{p,\max}(t) \cdot f_{i_t}(\theta_t) \cdot \mathbbm{1}[(i_t,\theta_t)  \in \usedcuts] \nonumber  \\
   &\leq  k^{p-1}\frac{c'_{p,\max}}{2^{r}} \cdot \sum_{t: \frac{c'_{p,\max}}{2^{r+1}}\leq c'_{p,\max}(t) \leq  \frac{c'_{p,\max}}{2^{r}}}f_{i_t}(\theta_t) \cdot \mathbbm{1}[(i_t,\theta_t)  \in \usedcuts].
\end{align*}
The first inequality is by H\"{o}lder's inequality (applied independently in each dimension in the computation of respective \(d_p\) and \(\| \cdot\|^p_p \) values).
The second inequality simply uses the condition of the summation.

\newcommand{\activeintervals}{\intervals_{\operatorname{act}}}

We now upper bound the expected value of \(\costinc(r)\). 
Let \(\activeintervals(r) \) be the set of dimension--interval pairs in \(\intervals\) that do not separate any pair of centers that are closer than \(c'_{p,\max}/(k^4 2^{r+1})\) in pseudo-distance but separate at least one pair of centers that are closer than  \(c'_{p,\max}/2^r\) in pseudo-distance. 
We prove the following lemma which is analogous to \cref{lemma:log-l1-drop} in \cref{sec:logk-upper-bound}.
\begin{lemma}
    For a round $r$, \(\E[\costinc(r)]\) is
    \[
    O\left(k^{p-1} \cdot \log k\right) \cdot \left( \sum_{(i, [a, b]) \in  \activeintervals(r)}  |b - a|^p \int_{a}^{b} P_{a,b}\left(\theta\right) f_i(\theta) d\theta \right).
    \] \label{lemma:lp-round-cost}
\end{lemma}
\begin{proof}
We consider ``trials'' of \(M\) consecutive iterations in round \(r\) where \[M =12 \cdot 2^{r+1} \ln(k) \cdot L_p/c'_{p,max}.\]
We perform independent trials until \(c'_{p,\max}(\cdot)\) at the end of a trial goes below \(c'_{p,\max}/2^{r+1}\).

Consider one trial and let \(s\) be the starting iteration of the trial.
Note that we have \[M \geq  3 \cdot 4 \cdot \ln(k) \cdot L_p/c'_{p,\max}(s)\] since \(c'_{p,\max}(s)\geq c'_{p,\max}/2^{r+1}\).
Thus, by \cref{lemma:logk-lp-drop}, after \(M\) iterations, round \(r\) ends with probability at least \(1-1/k \geq 1/2\). 
(Note that round \(r\) may end before all M iterations of a trial are completed.  
In such trials, we assume that we discard the additional cuts that are made after the round ends.)

Let \(\operatorname{UB_\ell} = \sum_{s=t'}^{t'+M-1}  f_{i_s}(\theta_s) \cdot \mathbbm{1}[(i_s,\theta_s)  \in \usedcuts] \) and observe that
\begin{align}
    \costinc(r) \leq k^{p-1} \frac{c'_{p,\max}}{2^{r}} \cdot \sum_{\ell} \operatorname{UB}_\ell \label{eq:costincwithub}
\end{align} where the sum is over all trials we perform in round \(r\).

We first upper bound each term \(\operatorname{UB}_\ell\) and then use the expectation of a geometric random variable to upper bound the expected value of \(\sum_{\ell}\operatorname{UB}_\ell\).
We have 
\begin{align}
\E[\operatorname{UB}_\ell] &\leq \sum_{s=t'}^{t'+M-1} \E\left[f_{i_s}(\theta_s) \cdot \mathbbm{1}[(i_s,\theta_s)  \in \usedcuts]\right] \nonumber \\
&\leq  \sum_{s=t'}^{t'+M-1} \frac{1}{L_p}  \sum_{(i, [a, b]) \in  \activeintervals(r)}  |b - a|^p \int_{a}^{b} P_{a,b}\left(\theta\right) f_i(\theta) d\theta. \label{eq:ubsum}
\end{align}

Note that we only sum over dimension--interval pairs in \(\activeintervals(r)\) as cuts made outside of this set will be discarded. 
To elaborate, the dimension--interval pair in which a cut \((i,\theta)\) is made can be outside of \(\activeintervals(r)\) for two reasons:
\begin{enumerate}
    \item Because it separates two centers that are closer than \(c'_{p,\max}/(k^42^{r+1}) \leq c'_{p,\max}(t')/k^4\). Then it will get discarded in \cref{line:discard}.
    \item Because it does not separate two centers that are closer than \(c'_{p,\max}/2^r\). Such a cut will not split any leaves in \cref{line:forp}.
\end{enumerate}
Consequently, for all \((i_s,\theta_s) \in \usedcuts\), we have \(\theta_s \in [a,b]\) for some interval \([a,b]\) such that \((i_s, [a,b]) \in \activeintervals(r)\).

Now, since the summed terms in \eqref{eq:ubsum} no longer depend on the summed index \(s\),  we now have
\begin{align*}
\E[\operatorname{UB}_\ell] &\leq \frac{M}{L_p} \sum_{(i, [a, b]) \in  \activeintervals} |b - a|^p \int_{a}^{b} P_{a,b}\left(\theta\right) f_i(\theta) d\theta \\
&= \frac{12 \cdot 2^{r+1} \ln(k)}{c'_{p,\max}} \sum_{(i, [a, b]) \in  \activeintervals} |b - a|^p \int_{a}^{b} P_{a,b}\left(\theta\right) f_i(\theta) d\theta.
\end{align*}

Now, considering that round \(r\) ends after a trial with probability at least \(1/2\), using the expected value of a geometric distribution, we conclude that
\[ \E \left[ \sum_{\ell}\operatorname{UB}_\ell \right] \leq \frac{24 \cdot 2^{r+1} \ln(k)}{c'_{p,\max}} \sum_{(i, [a, b]) \in  \activeintervals(r) } |b - a|^p \int_{a}^{b} P_{a,b}\left(\theta\right) f_i(\theta) d\theta. \]
The proof of the lemma then follows by combining this with the bound in \eqref{eq:costincwithub}.
\end{proof}

With \cref{lemma:lp-round-cost} in hand, we now prove \cref{thm:lp-modified-alg}.
\begin{proof}[Proof of \cref{thm:lp-modified-alg}]
    
Using \cref{lemma:lp-round-cost}, we can upper bound \(\sum_{r=0}^{\infty}\E[\costinc(r)]\) as follows:
\begin{align}
    \sum_{r=0}^{\infty}\, &\E[\costinc(r)] \nonumber \\
    &= O\left(k^{p-1} \cdot \log k\right) \cdot \sum_{r=0}^{\infty}  \sum_{(i, [a, b]) \in  \activeintervals(r)}  |b - a|^p \int_{a}^{b} P_{a,b}\left(\theta\right) f_i(\theta) d\theta \nonumber  \allowdisplaybreaks \\
&= O\left(k^{p-1} \cdot \log k\right) \cdot \sum_{(i, [a, b]) \in  \intervals} \sum_{r=0}^{\infty} \mathbbm{1}[(i, [a,b]) \in \activeintervals(r)] \cdot \left( |b - a|^p \int_{a}^{b} P_{a,b}\left(\theta\right) f_i(\theta) d\theta \right)  \nonumber \allowdisplaybreaks \\
&= O\left(k^{p-1} \cdot \log k\right) \cdot \sum_{(i, [a, b]) \in  \intervals}  |b - a|^p \int_{a}^{b} P_{a,b}\left(\theta\right) f_i(\theta) d\theta  \cdot  \left(\sum_{r=0}^{\infty} \mathbbm{1}[(i, [a,b]) \in \activeintervals(r)]\right). \nonumber \allowdisplaybreaks
\end{align}
We now claim that for any dimension--interval pair in \((i,[a,b]) \in \intervals\) 
\begin{align}
\sum_{r=0}^{\infty} \mathbbm{1}[(i, [a,b]) \in \activeintervals(r)] = \bigl\vert \left\{ r : (i,[a,b]) \in \activeintervals(r) \right\} \bigr\vert = O(\log k). \label{eq:rounds-for-interval}
\end{align}
Namely, fix some dimension--interval pair \((i,[a,b]) \in \intervals \). Let \(c\) be the smallest pseudo-distance between any pair of centers that are separated if a cut \((i,\theta)\) such that \(\theta \in [a, b]\) is made. 
Then \((i,[a,b])\) is in \(\activeintervals(r)\) only if \(c \leq c'_{p,\max}/2^r\) and \(c'_{p,\max}/2^{r+1} \leq k^4 c \), or equivalently, \( \log (c'_{p,\max}/{2 c k^4})  \leq r \leq \log (c'_{p,\max}/c) \) which yields \eqref{eq:rounds-for-interval}.
Thus we have
\begin{align}
\sum_{r=0}^{\infty}&\E[\costinc(r)] \nonumber \\
&= O\left(k^{p-1} \cdot \log^2 k\right) \cdot \sum_{(i, [a, b]) \in  \intervals}  \left( |b - a|^p \int_{a}^{b} P_{a,b}\left(\theta\right) f_i(\theta) d\theta \right) \nonumber \\
&= O\left(k^{p-1} \cdot \log^2 k\right) \cdot \sum_{(i, [a, b]) \in  \intervals}  \left( |b-a|^p \int_{a}^{b} p 2^{p-1} \frac{\min(\theta-a, b-\theta)^{p-1}}{|b-a|^p} f_i(\theta) d \theta \right) \nonumber \\
&= O\left(k^{p-1} \cdot \log^2 k\right) \cdot (p 2^{p-1}) \cdot \sum_{(i, [a, b]) \in  \intervals} \int_{a}^{b} \min(\theta-a, b-\theta)^{p-1} f_i(\theta) d \theta. \label{eq:expected_cost_inc_p}
\end{align}

Now to conclude the proof of \cref{thm:lp-modified-alg}, let  \( \cost'_p(\mathcal{U}) = \sum_{\bm x \in \mathcal{X}} d_p(\bm x, \closecenter(\bm x))\) which is the cost of \(\mathcal{U}\) defined in terms of the pseudo-distances. 
Recall that \( \|\bm x - \bm y\|^p_p \geq d_p(\bm x, \bm y) \) and hence we have  \(\cost_p(\mathcal{U}) \geq \cost'_p(\mathcal{U})\) where the equality holds if \(p=1\).
We then have
\begin{align}
   \cost_p(\mathcal{U}) &\geq \cost'_p(\mathcal{U}) = \sum_{\bm x \in \bm{\mathcal{X}}} d_p(\bm x, \closecenter(\bm x))  =  \sum_{\bm x \in \bm{\mathcal{X}}} \sum_{i \in [d]} \sum_{[a,b] \in \intervals_i(x_i,\closecenter(x)_i)} |a - b|^p \nonumber \\
    &= \sum_{\bm x \in \bm{\mathcal{X}}} \sum_{i \in [d]} \sum_{[a,b] \in \intervals_i(x_i,\closecenter(x)_i)} \int_a^b p (\theta - a)^{p-1} d \theta \nonumber \\
    &\geq  \sum_{\bm x \in \bm{\mathcal{X}}} \sum_{(i,[a,b]) \in \intervals}  \int_a^b p \cdot \min(\theta - a, b-\theta)^{p-1} \cdot \mathbbm{1}[\theta \text{ is between } x_i \text{ and } \closecenter(x)_i] \, d \theta \label{eq:explain} \\ 
    &=  p  \sum_{(i,[a,b]) \in \intervals}   \int_a^b \min(\theta - a, b-\theta)^{p-1} f_i(\theta) \, d \theta \label{eq:pcbound}.
\end{align}
    
The inequality in \eqref{eq:explain} above needs an explanation. First notice that, by the definition of \(\intervals\), summing over \((i, [a,b]) \in \intervals\) is the same as summing over \(i \in [d]\) and \([a,b] \in \intervals_i(\mu_i^-, \mu_i^+)\) .
Fix some point \(\bm{x}\) and dimension \(i\), and w.l.o.g. assume that
\(x_i \leq \closecenter(\bm x)_i\).
Then each interval \([a,b] \in \intervals_i(\mu_i^-, \mu_i^+)\) falls into one of the following categories:
\begin{enumerate}
    \item \(b \leq x_i \) or \(\closecenter(\bm x))_i \leq a\). 
    The contribution from such intervals are zero in both sides of the inequality in \eqref{eq:explain}.
    \item \(x_i \leq a \leq b \leq \closecenter(\bm x)_i \). 
    The contribution from such intervals to the left hand side of \eqref{eq:explain} is at least as the contribution to the right hand side because of the \(\min\) function. 
    \item \(a < x_i < b\). 
    Note that \(\int_a^b(\theta-a)^{p-1} d\theta = \int_a^b(b-\theta)^{p-1} d \theta\). 
    Hence, in this case, the contributions to both sides of \eqref{eq:explain} are equal if \(x_i \geq (a+b)/2 \). 
    Otherwise, the contribution to the L.H.S. is higher.
    \item \(a < \closecenter(\bm x) < b\). 
    This case is analogous to Item~3.
\end{enumerate}
The inequality in \eqref{eq:explain} follows by applying this observation to each interval in \(\intervals\) and each point in \(\bm{\mathcal{X}}\).

The theorem statement then follows by combining bounds \eqref{eq:expected_cost_inc_p} and \eqref{eq:pcbound}.
\end{proof}

\subsection{Implementation details}

Note that \cref{alg:randomp} differs from \cref{alg:random} in the sampling step and the new sample discarding step. Recall that the implementation details of 
\cref{alg:random} is presented in \cref{sec:implementation}.
In this section, assuming that we can sample a \(\theta \in [a,b]\) with p.d.f. \(P_{a,b}(\theta)\) from a given interval \([a,b]\) in \(\Theta(1)\) time, we show how to efficiently implement the sampling step of \cref{alg:randomp}.
In particular, we show how to select the dimension--interval pair \((i,[a,b])\) with probability proportional to \(|b-a|^p\). 
Once the cut is sampled, the discarding step can be implemented by simulating the splitting operation and ignoring the cut if it separates two centers that are too close.

Suppose that for each dimension \(i\), we maintain a data structure \(S_i\) that stores the intervals in \(\intervals_i(\mu_i^-, \mu_i^+)\) that are not yet split by a cut. 
There are \(k-1\) disjoint intervals in \(\intervals_i(\mu_i^-, \mu_i^+)\), and we assume they are ordered by the left coordinate and indexed \([a_1, b_1], \dots, [a_{k-1},b_{k-1}]\).
Additionally, each \(S_i\) supports the following operations:
\begin{enumerate}
    \item Initialize with all intervals in \(\intervals_i(\mu_i^-, \mu_i^+)\) in \(O(k \log k)\) time.
    \item Remove an interval in \(\intervals_i(\mu_i^-, \mu_i^+)\) in \(O(\log k)\) time.
    \item Given two indices \(\ell,r \in [k-1]\), answer the query for \(\sum_{j=\ell}^r |b_j - a_j|^p \mathbbm{1}[(a_j,b_j) \in S_i] \) in \(O(\log k)\) time.
\end{enumerate}
We can implement \(S_i\) as a segment tree. 

Now we can sample an interval from \(S_1, \dots, S_d\) as follows in \(O(d \log k)\) time: 
We first query \((1,k-1)\) in each tree, aggregate the results, and pick a dimension \(i\) with the correct probability.
Then we select an interval from \(S_i\) with the correct probability by employing a binary-search like algorithm.
To elaborate, we first query it for \((1,\lfloor (k-1)/2 \rfloor)\) and \((\lfloor (k-1)/2 \rfloor)+1, k-1)\) and use the results to randomly decide if the index of the sampled interval should be in the sub-range \(\{1,\dots,\lfloor (k-1)/2 \rfloor\}\) or \(\{\lfloor (k-1)/2 \rfloor)+1, \dots, k-1\}\).
Then we recursively apply the same procedure on the selected sub-range of indices until we end up with only one interval.
A crude runtime analysis gives \(O(\log^2 k)\) running time for the recursive sampling as there are \(O(\log k)\) queries and each query takes \(O(\log k)\) time. 
However we can modify the segment tree such that the partial sums maintained in the segment tree coincide with our queries so that each query can be answered in constant time.

\section{Lower bound}
\label{sec:lb}

In this section we show how to construct an instance of the clustering problem such that any explainable clustering has cost at least $\Omega(k^{p-1})$ times larger than the optimal non-explainable clustering for the objective function given by $\ell_p$ norm, for every $p \geq 1$. In particular, for $p=2$, this entails an $\Omega(k)$ lower bound for (explainable) $k$-means.

Let $m$ be a prime. Our hard instance is in $\R^d$ for $d = m\cdot(m-1)$ and the set of dimensions corresponds to the set of all linear functions over $\Z_m$ with non-zero slope. That is, we associate the $i$-th dimension with the function $f_i : x \mapsto (a_i x + b_i) \bmod m$, where $a_i = 1 + \lfloor i / m \rfloor$ and $b_i = i \bmod m$. Consider $k=m$ centers $\bm{\mu}^1, \dots, \bm{\mu}^k$ such that the $i$-th coordinate of the $j$-th center is given by $\mu^j_i = f_i(j)$. For each center $\bm{\mu}^j$ we create a set of $2d$ points $B_j$, each point differing from the center in exactly one dimension by either $-1$ or $+1$, i.e., $B_j = \{\bm{\mu}^j + c \cdot \bm{e}^i \mid c \in \{-1,1\}, i \in [d]\}$, where $\bm{e}^i$ denotes the standard basis vector in the $i$-th dimension. Then, our hard instance is just $\bigcup_{j \in [k]}B_j$.

Since every point is at distance exactly $1$ from its closest center, the cost of the optimal clustering $\mbox{OPT}$ is equal to the total number of points $n = 2dk$ (regardless of the $\ell_p$ norm). We prove that:
\begin{enumerate}[label={Claim \arabic*.},leftmargin=*]
  \item Any two centers are at the same distance $\delta=\Theta(d^{1/p}k)$ from each other.
  \item Any nontrivial threshold cut, i.e., one that separates some two centers, separates also some two points from the same $B_j$.
\end{enumerate}
It follows that, in any explainable clustering, already the first threshold cut (from the decision tree's root) forces some two points from the same set $B_j$ to eventually end up in two different leaves, and hence at least one of the $k$ leaves has to contain two points from two different $B_{j}$'s. The distance between these two points, by the triangle inequality, is at least $\delta - 2$, and therefore the cost of the explainable clustering is at least $\Omega(\delta^p) = \Omega(dk^p)$, which is $\Omega(k^{p-1}) \cdot \mbox{OPT}$.

\begin{proof}[Proof of Claim 1]
Fix two different centers $\bm{\mu}^{j_1}$, $\bm{\mu}^{j_2}$, $j_1 \neq j_2$. Their distance $\delta$ satisfies
\[\delta^p = 
\sum_{i \in [d]} \big|f_i(j_1) - f_i(j_2)\big|^p =
\sum_{a=1}^{m-1} \sum_{b=0}^{m-1} \big|(a j_1 + b) \bmod m - (a j_2 + b) \bmod m \big|^p.\]
For $a \in \{1,\ldots,p-1\}$, let $x(a) = (a j_1 - a j_2) \bmod m$. Observe that
\[\big|(a j_1 + b) \bmod m - (a j_2 + b) \bmod m \big| \in \{x(a), m-x(a)\},\]
and whether it is $x(a)$ or $m-x(a)$ depends on $b$, with it being $x(a)$ for exactly $m-x(a)$ values of $b$ and $m-x(a)$ for the remaining $x(a)$ values of $b$. Hence,
\[\delta^p = \sum_{a=1}^{m-1} (m-x(a)) \cdot x(a)^p + x(a) \cdot (m-x(a))^p.\]
Since $j_1 \not\equiv j_2 \pmod{m}$, we have $\big\{x(a) \mid a \in \{1,\ldots,m-1\}\big\} = \{1, \ldots, m-1\}$, and
\[\delta^p = \sum_{i=1}^{m-1} (m-i) \cdot i^p + i \cdot (m-i)^p = 2 \cdot \sum_{i=1}^{m-1} (m-i) \cdot i^p = \Theta(m^{p+2}) = \Theta(dk^p).\]
\end{proof}

\begin{proof}[Proof of Claim 2]
Let the cut be $(i, \theta)$.
It must be that $0 \le \theta < m-1$, because otherwise the cut would not separate any two centers. Note that there exists a center $\bm{\mu}^j$ with $\bm{\mu}^j_i = \lfloor \theta \rfloor$. Indeed, consider $j = (\lfloor \theta \rfloor - b_i) \cdot a_i^{-1} \bmod m$, using the fact that $a_i$ and $m$ are coprime. To finish the proof observe that the cut separates point $(\bm{\mu}^{j} + \bm{e}^i) \in B_{j}$ from all other points in $B_{j}$.
\end{proof}

\bibliographystyle{plainurl}
\bibliography{refs}

\clearpage
\appendix

\section{The minimum cut algorithm loses \texorpdfstring{$\Omega(k)$}{Ω(k)} factor for \texorpdfstring{$k$}{k}-medians}
\label{example-k}

We give an example where the minimum cut algorithm of~\cite{moshkovitz} produces a threshold tree with cost $\Omega(k)$ times the cost of an optimal clustering in the $\ell_1$-norm.
The idea is to start with the lower bound example in Section~\ref{sec:lb} since any two centers are ``far apart''.
By adding a dimension for each center in which fewer edges are cut, the minimum cut will make linearly many cuts that split only one center. 
Combined with the large distance to reassign a point to the wrong center, the result is the minimum cut algorithm losing an $\Omega(k)$ factor.
In the $\ell_1$ norm, it suffices to map half of the coordinate values to -1 and the other half to +1 and still maintain the ``large'' distance between centers.
The remainder of this section is a formal description of the instance.

Take the lower bound example from Section~\ref{sec:lb} and increase the dimension by $k$.
Now the points are in $\R^{d+k}$ with $d + k$ coordinates (recall that $d = m(m-1)$ and $k = m$ with $m$ prime).
First, we describe the $k$ centers $\U' = \{\bm{\mu'}^1, \dots, \bm{\mu'}^k\}$ as a mapping from the centers $\U = \{\bm{\mu}^1, \dots, \bm{\mu}^k\}$ in Section~\ref{sec:lb}.
For the first $d$ coordinates, $\mu'^j_i = \mu^j_i \mod 2$. 
For the last $k$ coordinates, center $\bm{\mu'}^j$ has a 0 in every coordinate $d + i$, $1 \leq i \leq k$, except coordinate $d+j$ which is a 1.

The reasoning behind this mapping is that the family of functions $f_i$ in Section~\ref{sec:lb} is the standard construction of a family of pairwise independent hash functions~\cite{Fredman84}.
In particular, if $f_{a,b}(x) = (a x + b) \bmod m$ and $h_{a,b}(x) = f_{a,b}(x) \bmod 2$, then for $x \neq y$, $h_{a,b}(x) = h_{a,b}(y)$ with probability at most $1/2$ when $a$ and $b$ are chosen uniformly at random from $\{0,1, \dots, m-1\}$, $a \neq 0$.
Recall that $f_i(x) = (a_i x + b_i) \bmod m$ where $a_i$, $b_i$ range over all elements in $\{1, 2, \dots, m-1\}$, $\{0, \dots, m-1\}$, respectively, and $\mu_i^j = f_i(j)$.
Fix any pair of centers $\bm{\mu}^{j_1}, \bm{\mu}^{j_2}$ where $j_1 \neq j_2$.
Note that picking $i \in [d]$ uniformly at random is equivalent to picking $a_i, b_i \in \{0,1, \dots, m-1\}$, $a \neq 0$, uniformly at random due to the definition of $a_i, b_i$.
We have $\mu'^{j_1}_i = (\mu^{j_1}_i \bmod 2) = (\mu^{j_2}_i \bmod 2) = \mu'^{j_2}_i$ with probability at most $1/2$ over the uniformly random choice of $i$, so any pair of centers are the same on at most $1/2$ of the coordinates.
Hence, our new centers $\U'$ are at pairwise distance $\Theta(d)$.

Now we define the remaining points. 
Let $\bm{e}^i$ be the standard $(d+k)$-dimensional $i$-th basis vector.
Similar to Section~\ref{sec:lb}, we have a set $B_j'$ for each center $\bm{\mu}'^j$ with $2d$ points where each point differs from $\bm{\mu}'^j$ on one of the first $d$ coordinates by $\pm 1$. Additionally, we want $(k-1)/2$ points to differ on one of the last $k$ coordinates.
To this end, define $B_j' = \{\bm{\mu'}^j + c \cdot \bm{e}^i \mid c \in \{-1,1\}, i \in [d]\} \cup \{ \bm{\mu'}^j - \bm{e}^{d+j}\}^{(k-1)/2}$, where $\{ \bm{\mu'}^j - \bm{e}^{d+j}\}^{(k-1)/2}$ denotes a multiset of $(k-1)/2$ copies of the point $\bm{\mu'}^j - \bm{e}^{d+j}$.

In particular, our construction has the following properties:
\begin{enumerate}
    \item The distance between any pair of centers is $\Theta(d)$.
    \item A cut $(i, \theta)$ in any dimension $i$, $1 \leq i \leq d$, and $\theta \in (0, 1)$ splits some two centers and the number of points separated is equal to the number of centers.
    \item A cut $(i, \theta)$ in any dimension $i$, $d+1 \leq i \leq d+k$, and $\theta \in (0, 1)$ splits some two centers and separates $\approx k/2$ points.
\end{enumerate}
Property (2) holds because for any dimension $1 \leq i \leq d$ and for each center $\bm{c}$, $c_i$ is either 1, in which case there is exactly one point at $\bm{c} - \bm{e}^i$, or 0, in which case there is exactly one point at $\bm{c} + \bm{e}^i$.
Note that this further implies that, when (after separating some centers) $x$ centers are remaining, the number of points separated by a cut of type (2) will be equal to $x$.  
Then the cuts in (3) will be minimum for $\approx k/2$ cuts of all minimum cuts required to separate all centers since each cut in (3) separates exactly one center from the remaining centers.
Hence we have that the minimum cut algorithm of~\cite{moshkovitz} will construct a threshold tree with $\Omega(k)$ height by making some $\Omega(k)$ cuts in dimensions $d+1$ through $d+k$.

To see that the minimum cut algorithm loses a $\Omega(k)$ factor, note that the optimal clustering has a value of $2dk + (k-1)k/2 = \Theta(k^3)$.
The first term in the sum is because each of $2d$ points in each cluster differs from the center by $\pm 1$ in exactly one of the first $d$ coordinates and the second term is because $(k-1)/2$ of the points in each cluster differ by $-1$ from the center in one of the last $k$ coordinates.
On the other hand, an algorithm that always makes a minimum cut incurs a cost of $\Theta(dk^2)$ to reassign $\approx k/2$ points to the wrong center for $\approx k/2$ centers, just for those cuts of type (2).
This gives an overall cost of $\Omega(d k^2)$ for the threshold tree produced.
Since $d = \Theta(k^2)$ we have that the minimum cut algorithm is $\Omega(k)$ away from the cost of an optimal clustering.

\section{Omitted proofs of Section~\ref{ell-1}} %
\label{sec:omitted-sec-3}

\subsection{Upper bounding cost increase of a round}
\label{appendix:cost-increase}

Here we give the formal proof of~\cref{lemma:cost-increase}, restated below.
Recall that
  \[\costinc(r) = \sum_{t\,:\, c_{\max}(t) \in (c_{\max}/2^{r+1}, c_{\max}/2^r]} c_{\max}(t) f_{i_t} (\theta_t)
  \]is the random variable that upper bounds the cost increase caused by the cuts selected during the iterations $t$ when $c_{\max}/2^{r+1} < c_{\max}(t) \leq c_{\max}/2^r$.

\costincrease*

\begin{proof}
      Let $t$ be the first iteration when $c_{\max}(t) \leq c_{\max}/2^r$ and let $M = 3 \ln(k) \cdot {2L}/{c_{\max}(t)}$ as in Lemma~\ref{lemma:l1-drop}. In the following, we use $\cost_M$ to denote the random variable that equals the cost increase caused by adding $M$ uniformly random cuts after the $t$-th iteration.  Then 
  \begin{align*}
      \E[\cost_M] \leq M \cdot c_{\max}(t)\cdot \E_{(i, \theta)}[f_i(\theta)] 
    \leq M  \cdot c_{\max}(t) \cdot \cost_1(\mathcal{U})/L
     = 6 \ln(k) \cdot \cost_1(\mathcal{U})\,,
  \end{align*}
    where the first inequality holds because $c_{\max}(t)$ is monotonically decreasing and the second inequality is by Lemma~\ref{lemma:l1-nr-cut}.
  At the same time, if we let $H$ denote the event that $c_{\max}(t)$ has halved after adding these $M$ cuts, i.e., that $c_{\max}(t+M) \leq c_{\max}(t)/2$, then $\Pr[H] \geq 1-1/k$ by Lemma~\ref{lemma:l1-drop}.
  We now upper bound the expectation of  $\costinc(r)$  by considering ``trials'' of $M$ cuts until one of these succeeds in halving $c_{\max}(t)$. Indeed,  split the sequence of random cuts selected by the algorithm after iteration $t$ into such trials $A_1, \ldots, A_\ell$ where each $A_j$ consist of $M$ cuts, and $A_\ell$ is the first successful trial in the sense that selecting (only) those cuts after iteration $t$ would cause $c_{\max}(t)$ to halve. Then we must have that $c_{\max}(t)$ has halved also after adding all the cuts in the $\ell$ trials. It follows that $\costinc(r)$ is upper bounded by the cost increase caused by the cuts in $A_1, A_2, \ldots, A_\ell$. We can thus upper bound $\E[\costinc(r)]$ by the expected cost of these trials until one succeeds:
  \[
    \sum_{i=0}^{\infty} \Pr[H] \cdot \Pr[\neg H]^i \cdot \left( \E\left[\cost_M \mid H \right] +  i \cdot \E\left[\cost_M \mid \neg H \right] \right)\,,
  \]
  where we use $\E[\cost_M \mid H]$ and $\E[\cost_M \mid \neg H]$ for the expected costs of a successful and unsuccessful  trials, respectively. By standard calculations (as for the geometric distribution), this upper bound simplifies to
$\E\left[\cost_M \mid H \right] + \frac{\Pr[\neg H]}{\Pr[H]}\E\left[\cost_M \mid \neg H \right]$. This can be further rewritten as
  \[
    \frac{1}{\Pr[H]} \cdot \left( \Pr[H] \cdot \E[\cost_M \mid H] + \Pr[\neg H] \cdot \E[\cost_M \mid \neg H]\right)  = \frac{\E[\cost_M]}{\Pr[H]} \leq 12 \ln(k) \cdot \cost_1(\mathcal{U})\,,
  \]
  where we used that $\Pr[H] \geq 1-1/k \geq 1/2$.
\end{proof}

\subsection{Theorem~\ref{thm:l1-modified-alg} implies furthermore statement of Theorem~\ref{thm:l1-alg}}
\label{sec:thm2-implies-thm1}

Recall that the difference between Algorithm~\ref{alg:random} and the modified version is that Algorithm~\ref{alg:random} samples cuts uniformly at random whereas the modified version only adds a random cut if it does \emph{not} separate two centers that are within distance $c_{\max}(t) /k^4$. 

Algorithm~\ref{alg:random} adds $k-1$ cuts to its tree. We now argue that these $k-1$ cuts are with probability at least $1-1/k$ sampled from the same distribution as the $k-1$ cuts added by the modified version. This then implies the furthermore statement of Theorem~\ref{thm:l1-alg} since Theorem~\ref{thm:l1-modified-alg} says that the expected cost of the modified algorithm is $O(\log^2 k) \cdot \cost_1(\mathcal{U})$.   To this end, consider the  $i$-th such cut and let $t$ be the iteration when the $(i-1)$-st cut was added to the tree. Then when the $i$-th cut is added there must be two centers in the same leaf at distance $c_{\max}(t)$. So the probability that two centers within distance $c_{\max}(t) /k^4$ are separated by the $i$-th cut  (which is a uniformly random cut among all cuts that would separate at least two centers in the same leaf) is at most $1/k^4$. There can be at most $\binom{k}{2}$ such pairs and so by the union bound,  we can conclude that, with probability at least $1-1/k^2$, the $i$-th cut of Algorithm~\ref{alg:random} does not separate any such nearby centers. We can thus  view the distribution from which Algorithm~\ref{alg:random} samples the $i$-th cut as follows: With probability $p \leq 1/k^2$ it samples a uniformly random cut that separates two centers within distance $c_{\max}(t) /k^4$ and with remaining probability it samples a uniformly random cut that does not separate any such centers, i.e., from the same distribution that the modified algorithm samples the $i$-th cut from. Applying the union bound over the $k-1$ cuts then yields the furthermore  statement of Theorem~\ref{thm:l1-alg}. Finally, we remark that the same arguments imply a larger success probability if applied to the modified algorithm that only adds cut that do not separate centers within distance $c_{\max}(t)/k^\ell$ for some $\ell \geq 4$.

\end{document}